\documentclass
[reqno]{amsart}
\usepackage[utf8]{inputenc}
\usepackage{amssymb,amscd}
\usepackage{thmtools}

\usepackage[textsize=tiny]{todonotes}
\usepackage{hyperref}
\usepackage[shortlabels]{enumitem}
\usepackage{float} 
\usepackage{cleveref} 
\usepackage{graphicx}
\usepackage{float}

\usepackage{mathtools} 
\usepackage{url}
\usepackage{cases} 
\usepackage{mathrsfs}

\newtheorem{definition}{Definition}

\newtheoremstyle{alstandard}{7pt}{3pt}{\rm}{}{\scshape}{:}{0.5em}{}
\theoremstyle{alstandard}
\swapnumbers
\declaretheorem[name=Theorem]{theorem}
\numberwithin{theorem}{section}
\declaretheorem[sibling=theorem, name=Lemma]{lemma}

\declaretheorem[sibling=theorem, name=Definition]{defi}

\usepackage{pifont}
\begin{document}
	\title[]{Parasitic actuation lag limits the minimum employable time headway in connected and autonomous vehicles} 
	
	\author[]{Guoqi Ma, Prabhakar R. Pagilla, and Swaroop Darbha}
	\address{Department of Mechanical Engineering, Texas A\emph{\&}M University, College Station, TX 77843, USA}
	\email{gqma@tamu.edu, ppagilla@tamu.edu, dswaroop@tamu.edu}

	\begin{abstract} 
Adaptive and cooperative adaptive cruise control (ACC and CACC) and next generation CACC (CACC+) systems usually employ a constant time headway policy (CTHP) for platooning of connected and autonomous vehicles (CAVs). In ACC, the ego vehicle uses onboard sensors to measure the position and velocity of the predecessor vehicle to maintain a desired spacing. The CACC and CACC+ systems use additional information, such as acceleration(s) communicated through vehicle-to-vehicle (V2V) communication of the predecessor vehicle(s); these systems have been shown to result in improved spacing performance, throughput, and safety over ACC. Parasitic dynamics are generally difficult to model and the parasitic parameters (delay, lag, etc.) are difficult to obtain. Parasitic actuation delays can have deleterious effects and impose limits on the mobility and safety of CAVs. It is reasonable to assume that the bounds on parasitic actuation delays are known a priori. For CAVs, we need to address both internal stability and string stability in the presence of parasitic actuation delays. This requires robustness of internal and string stabilities for all values of parasitic actuation delays that are within the specified upper bound. In this paper, we show the following: given an upper bound on the parasitic actuation delays as $\tau_0$, the minimum employable time headway for ACC is $2\tau_0$, for CACC it is $2\tau_0/(1+k_a)$ where $k_a \in [0,1)$ is the control gain associated with the communicated acceleration of the predecessor vehicle, and $4 \tau_0/((1 + r) (1 + r k_a))$ for CACC+ (`$r$' predecessors look-ahead) where $k_a \in [0, 1/r)$ is the control gain associated with the communicated accelerations of the $r$ predecessor vehicles. The inclusion of the internal stability in the string stability condition is analyzed based on Pontryagin's interlacing theorem for time delay systems. We provide comparative numerical results to corroborate the achieved theoretical findings.
	\end{abstract}

	\maketitle
	\raggedbottom
	
\section{Introduction}\label{sec:intro}
It is well known that connectivity and coordination through onboard sensing and/or communication of autonomous vehicles improve road safety and traffic throughput~\cite{SHETTY2021103133}. In earlier designs and algorithms in Connected and Autonomous Vehicles (CAVs), the control objective was mainly focused on maintaining a desired following distance between vehicles to improve safety and capacity; this was achieved through Conventional Cruise Control (CCC) and Adaptive Cruise Control (ACC) by employing vehicle onboard sensing to measure the position and velocity of the predecessor vehicle. In recent years, with substantial advances in wireless communication technologies, such as Dedicated Short-Range Communications (DSRC)~\cite{5888501}, Long Term Evolution (LTE)~\cite{8108463}, Vehicle-to-Vehicle (V2V) communication~\cite{vinel2015vehicle}, the 5th generation communication system (5G)~\cite{9345798}, it has become easier to communicate and exchange information (e.g. acceleration/deceleration) between vehicles; this capability enables advanced vehicle following systems employing such as Cooperative Adaptive Cruise Control (CACC)~\cite{8569947} and next generation CACC (CACC+)~\cite{darbha2018benefits} and realizes the potential to further improve safety and increase throughput.

In ACC, CACC and CACC+ systems, the inter-vehicular spacing policy specifies the desired following distance between vehicles, which, in general, consists of the Constant Spacing Policy (CSP) and the Variable Spacing Policy (VSP)~\cite{iet2020platooning}. In CSP, the desired following distance between vehicles is a constant, while in VSP the desired following distance is varying to overcome some drawbacks with CSP and render enough safety threshold by taking into account the current motion status of the platoon. A commonly adopted VSP is the Constant Time Headway Policy (CTHP)~\cite{swaroop2001review}, where the desired following distance between vehicles is a linear function of the velocity of the host vehicle and the proportional coefficient is the time headway which can influence the spacing between vehicles and the tightness of the platoon significantly. A basic design requirement on CAVs is string stability under a certain inter-vehicular spacing policy, that is, under the chosen policy, the spacing errors are not amplified when propagating along the platoon string. Existing work on CAVs focuses on a wide range of issues including experimental validation~\cite{GE2018445,5571043,qin2019experimental,di2019design}, communication mechanisms~\cite{bazzi2017performance,10.1115/1.4036565,vegamoor2021string,7055887,fiengo2019distributed}, mixed human and autonomous vehicle platoons~\cite{9246221,jiang2021robust,garg2023can,MAHBUB2023111115}, intersection coordination~\cite{ILGINGULER2014121,XU2018322,8082807}, merging or cut-in analysis~\cite{SCHOLTE2022103511,li2024disturbances}, motion planning~\cite{muller2022motion,liu2022markov}, communication resource allocation~\cite{9885733}, safety-critical control~\cite{zhao2023safety}, traffic signal control~\cite{li2024cooperative} and so on. A more comprehensive review on the development of CAVs is provided by~\cite{Vegamoor_review_paper,FENG201981}.

In recent years, the time headway minimization problem in CTHP has received increasing attention with the goal of seeking the smallest achievable lower bound of the time headway such that the vehicle platoon can achieve the prescribed string stability with the tightest possible spacing. In~\cite{darbha2018benefits,BIAN201987,10038652}, the problem of the lower bound of the time headway for string stable CAVs has been investigated by assuming a first-order parasitic actuation lag. In~\cite{bekiaris2023robust}, the predictor-based control method was adopted to compensate the delay effect for the vehicle dynamics model with general actuation delay by adding integral terms in the control input; however, the points of departure from our work are that the time delay has to be known in this paper for the predictor to work, although robustness to small changes in time delay was considered, there is no explicit lower bound on the time headway as a function of time delay. It should be noted that in~\cite{klinge2009time,ploeg2014graceful}, the influence of time headway on string stability was investigated; however, a clear lower bound of the implementable time headway was not provided from a theoretical point of view. Although a comparison between different vehicle dynamics models for a certain brand of vehicles has been made through experimental analysis~\cite{MA2022103927}, there is no known theoretical analysis on obtaining the time headway lower bound for the general parasitic lag to achieve string stability. 

The contribution of this paper is the following. Given a maximum possible parasitic actuation delays of $\tau_0$, we show that the lower bounds of the minimum employable time headway are, respectively: (1) $2 \tau_0$ for ACC, (2) $\frac{2 \tau_0}{1 + k_a}$ for CACC, where $k_a \in[0, 1)$ is the control gain associated with the acceleration of the predecessor vehicle, and (3) $\frac{4 \tau_0}{( 1 + r ) ( 1 + r k_a )}$ for CACC+, where $r$ is the number of predecessors whose motion information is utilized in the ego vehicle and $k_a \in [0,\frac{1}{r})$ is the control gain associated with the accelerations of the predecessor vehicles. Our approach for internal stabilization relies on the Pontryagin's interlacing theorem for time delay systems \cite{bhattacharyya2018linear}.

The remainder of the paper is organized as follows. Section~\ref{section:problem-formulation-and-preliminaries} contains preliminaries that include vehicle dynamics and relevant definitions. The main theoretical results are developed in Section~\ref{section:main-results}. An illustrative numerical example and simulation results are provided in Section~\ref{section:numerical-simulations} to corroborate the achieved main theoretical results. Finally, some concluding remarks are given in Section~\ref{section:conclusion}.
	\section{Preliminaries}\label{section:problem-formulation-and-preliminaries}
\subsection{Vehicle string model}

Consider a string of autonomous vehicles equipped with V2V communication as illustrated in Figure~\ref{fig_illustration_of_vehicle_platoon}, in which $V_0$ denotes the lead vehicle, $V_1, V_2, \cdots, V_{N - 1}$, and $V_{N}$ denote the following vehicles. 
\begin{figure}[!htb]
\centering{\includegraphics[scale=0.55]{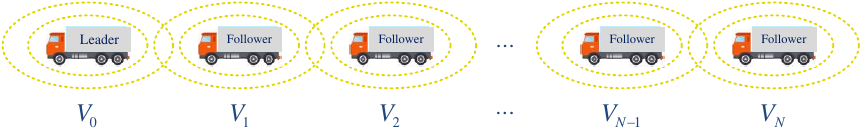}}
\caption{An illustration of a connected and autonomous vehicle platoon with V2V communication.
\label{fig_illustration_of_vehicle_platoon}}
\end{figure} 
A commonly used longitudinal double integrator dynamics model of the $i$-th following vehicle in the string (assuming feedback linearization of the nonlinear vehicle dynamics \cite{Vegamoor_review_paper}) is given by
\begin{align} \label{eq:vehicle-dynamics}
\left\{\begin{array}{*{20}{l}}
\dot{x}_{i}(t) = v_{i}(t), \\
\dot{v}_i(t) = a_i(t), \\
\tau \dot{a}_i(t) + a_{i}(t) = u_{i}(t),
\end{array}
\right.
\end{align} 
where $x_{i}(t)$, $v_i(t)$, $a_{i}(t)$, $u_{i}(t)$ represent the position, velocity, acceleration, and synthetic control input of the $i$-th following vehicle at time instant $t$, respectively, $\tau$ denotes the parasitic actuation lag, and the index $i \in \mathscr{N} = \{ 1, 2, \cdots, N \}$, where $N$ is the total number of the following vehicles in the platoon. 

The first-order parasitic actuation lag in \eqref{eq:vehicle-dynamics} models the physical limitation of the actuator of not providing instantaneous actuation. The first-order model is an approximation of the general form of actuator delay given by 
\begin{align} \label{eq:actual-acceleration-model}
 a_i(t) = u_i(t - \tau). 
\end{align} 
As a result, the vehicle dynamics model that incorporates the general actuation delay is given by
\begin{align} \label{eq:actual-vehicle-dynamics}
 \begin{cases}
 \dot{x}_i(t) = v_i(t), \\
 \dot{v}_i(t) = a_i(t), \\
 a_i(t) = u_i(t - \tau).
 \end{cases}
\end{align}
It is assumed that $\tau$ is {\it uncertain} with $\tau \in \left( 0, \tau_{0} \right]$, where $\tau_{0}$ is a positive constant.

\subsection{Definitions and background results}

\begin{definition}
 Let $e_i(t) = x_{i}(t) - x_{i - 1}(t) + d$ be the spacing error with CSP for the $i$-th following vehicle, where $d$ is the minimum or standstill spacing between adjacent vehicles, the generalized or velocity-dependent inter-vehicular spacing error with CTHP for the $i$-th following vehicle is defined as follows: 
\begin{align} \label{eq:delta-i-definition}
 \delta_i(t) = e_i(t) + h_w v_i(t),   
\end{align}
where $h_w$ is the time headway. 

\end{definition}

\begin{defi}[\cite{konduri2017robust}]
 The connected and autonomous vehicle platoon with vehicle dynamics model~\eqref{eq:actual-vehicle-dynamics} is said to be robustly string stable if $\forall \tau \in (0, \tau_0]$ it holds that $\Vert H(s;\tau) \Vert_{\infty} \le 1$, where $H(s;\tau)$ is the inter-vehicular spacing error propagation transfer function that satisfies $\hat{\delta}_i(s) = H(s;\tau) \hat{\delta}_{i - 1}(s)$, in which $\hat{\delta}_i(s)$ denotes the Laplace transform of $\delta_i(t)$.    
\end{defi}

We will employ the following two lemmas to show internal stability.

\begin{lemma}[\cite{bhattacharyya2018linear}]  (cf. theorem 3.3 on page 72 therein) \label{lemma:internal-stability} Consider the characteristic polynomial 
 \begin{align}
 D(s) = n_0(s) + \sum\limits_{i = 1}^m e^{-s \ell_i} n_i(s), 
 \end{align}
 that satisfies:
 \begin{enumerate}
     \item ${\rm deg}[n_0(s)] = q$ and ${\rm deg}[n_i(s)] \le q$, for $i = 1, 2, \cdots, m$;
      \item $0 < \ell_1 < \ell_2 < \cdots < \ell_m$; 
    \item $\ell_i = \alpha_i \ell_1$, $i = 2, \cdots, m$, and $\alpha_i$ are non-negative integers (commensurate delays).
 \end{enumerate}
Also, consider the following quasi-polynomial: 
 \begin{align}
  D^{\ast}(s) &= e^{ s \ell_m } D(s) \nonumber \\
  &= e^{s \ell_m } n_0(s) + \sum\limits_{i = 1}^m e^{s (\ell_m - \ell_i )} n_i(s).
 \end{align} 
 Substituting $s = j \omega$, $D^{\ast}(j \omega)$ can be written in the following form:
 \begin{align}
  D^{\ast}(j \omega) = D_r(\omega) + j D_i(\omega).   
 \end{align}
 Then, $D(s)$ or $D^{\ast}(s)$ is stable if and only if 
 \begin{enumerate}
     \item $D_r(\omega)$ and $D_i(\omega)$ have only simple, real roots and these interlace; 
    \item $D_i^{\prime}(\omega_0) D_r(\omega_0) - D_i(\omega_0) D_r^{\prime}(\omega_0) > 0$ for some $\omega_0$ in $(-\infty, \infty)$, where $D_i^{\prime}(\omega)$ and $D_r^{\prime}(\omega)$ denote the first derivative with respect to $\omega$ of $D_i(\omega)$ and $D_r(\omega)$, respectively.
 \end{enumerate}

\end{lemma}

\begin{lemma} [\cite{bhattacharyya2018linear}] (cf. theorem 3.4 on page 75 therein) \label{lemma:real-roots}
Let $M$ and $N$ denote the highest powers of $s$ and $e^{s}$, respectively, in the quasi-polynomial $D^{\ast}(s)$, and let $D^{\ast}(j \omega) = D_r(\omega) + j D_i(\omega)$. Let $\eta$ be an appropriate constant such that the coefficients of terms of the highest degree in $D_r(\omega)$ and $D_i(\omega)$ do not vanish at $\omega = \eta$. Then for the equations $D_r(\omega) = 0$ or $D_i(\omega) = 0$ to have only real roots, it is necessary and sufficient that in each of the intervals 
\begin{align*}
 - 2 l \pi + \eta \le \omega \le 2 l \pi + \eta, l = l_0, l_0 + 1, l_0 + 2, \cdots, 
\end{align*} 
 $D_r(\omega)$ or $D_i(\omega)$ has exactly $(4 l N + M)$ real roots for a sufficiently large $l_0$.

\end{lemma}

    \section{Main Results} \label{section:main-results}

\subsection{Control input and inter-vehicular spacing error propagation equation for ACC and CACC}

For the $i$-th following vehicle in the platoon, consider the following ACC/CACC input: 
\begin{align} \label{eq:control-input-cacc-case}
 u_i(t) = k_a a_{i - 1}(t) - k_v ( v_i(t) - v_{i - 1}(t) ) - k_p \delta_{i}(t), 
\end{align}
 where $k_v$, $k_p$ are positive and $k_a$ is nonnegative control gains to be designed, with $k_a=0$ for ACC and $k_a\neq 0$ for CACC.  Substituting~\eqref{eq:control-input-cacc-case} into~\eqref{eq:actual-vehicle-dynamics} and simplifying, the inter-vehicular spacing error propagation equation can be obtained as 
\begin{align} \label{eq:inter-vehicular-spacing-error-propagation-equation-cacc-case}
 \hat{\delta}_i(s) = H(s;\tau) \hat{\delta}_{i - 1}(s), 
\end{align}
 where $\hat{\delta}_i(s)$ is the Laplace transform of $\delta_i(t)$, $H(s;\tau) = \frac{\mathcal{N}(s)}{\mathcal{D}(s)}$, in which 
\begin{align*}
 \mathcal{N}(s) = k_a s^2 + k_v s + k_p, \ \mathcal{D}(s) = s^2 e^{\tau s} + \gamma s + k_p,
\end{align*}
where $\gamma = k_v + h_w k_p$.

\subsubsection{String stability analysis}

\begin{theorem} \label{theorem:string-stability-condition-ka-range}
For the platoon given by the inter-vehicular spacing error propagation equation~\eqref{eq:inter-vehicular-spacing-error-propagation-equation-cacc-case} to be robustly string stable, 
it is necessary that 
 \begin{align}\label{eq:ka_upper_bound}
  k_a < 1.   
 \end{align}
\end{theorem}
 \begin{proof}
 Considering $H(s;\tau)$ defined in~\eqref{eq:inter-vehicular-spacing-error-propagation-equation-cacc-case}, we can first compute 
 \begin{align}
  &\vert H(j \omega; \tau) \vert^2 \nonumber \\
  &\quad = \frac{ k_a^2 \omega^4 + ( k_v^2 - 2 k_a k_p ) \omega^2 + k_p^2 }{ \omega^4 + \gamma^2 \omega^2 - 2 \gamma \omega^3 \sin(\tau \omega) + k_p^2 - 2 k_p \omega^2 \cos(\tau \omega) } \nonumber \\
  &\quad \ge \frac{ k_a^2 \omega^4 + ( k_v^2 - 2 k_a k_p ) \omega^2 + k_p^2 }{ \omega^4 + \gamma^2 \omega^2 + 2 \gamma \omega^3 + k_p^2 + 2 k_p \omega^2 }.
 \end{align}
 We will then conduct the proof by contradiction through the following two cases: \\
 \noindent Case 1: If $k_a > 1$, then for large $\omega$:
 \begin{align}
\lim\limits_{\omega \to \infty} \vert H(j \omega;\tau)
 \vert^2  
& \ge \lim\limits_{\omega \to \infty} \frac{k_a^2 \omega^4 + (k_v^2 - 2 k_a k_p) \omega^2 + k_p^2 }{ \omega^4 + 2 \gamma \omega^3 + (\gamma^2 + 2 k_p) \omega^2 + k_p^2 } \nonumber \\
 & = k_a^2 > 1.   
 \end{align}
This implies $\|H(j\omega;\tau)\|_{\infty} > 1$. Thus, $k_a \le 1$. \\
 
\noindent Case 2: If $k_a = 1$, then we have 
{\small
\begin{align*}
\vert H(j \omega; \tau) \vert^2 
= \frac{ \omega^4 + \left[k_v^2 - 2 k_p \right] \omega^2 + k_p^2 }{ \omega^4 + \left[ \gamma^2 - 2 \gamma \omega \sin(\tau \omega)  - 2 k_p \cos(\tau \omega) \right] \omega^2 + k_p^2}. 
\end{align*}
}
We will show that $\forall \tau \in (0, \tau_0]$, there exists an $\hat{\omega}$ such that $\tau\hat{\omega} = 2k\pi + \frac{\pi}{2}$ for some sufficiently large positive integer $k$ and $\|H(j\omega;\tau)\|_{\infty} > 1$. 
Now, choose an integer $k$ large enough such that 
\begin{align}\label{eq:k-bound}
 k > \frac{(\gamma^2 - k_v^2 + 2 k_p) \tau}{4 \pi \gamma} - \frac{1}{4}.
\end{align}
Define $\hat{\omega}$ as follows: 
\begin{align} \label{eq:omega-ka1}
\hat{\omega} &:= \frac{\pi}{2 \tau} + \frac{2 k \pi}{\tau}.
\end{align}
Note that
{\small
\begin{align*}
    \|H(j\omega;\tau)\|_{\infty} \ge \left| \frac{ \hat{\omega}^4 + \left[k_v^2 - 2 k_p \right] \hat{\omega}^2 + k_p^2 }{ \hat{\omega}^4 + \left[ \gamma^2 - 2 \gamma \hat{\omega} \sin(\tau \hat{\omega})  - 2 k_p \cos(\tau \hat{\omega}) \right] \hat{\omega}^2 + k_p^2} \right|.
\end{align*}
}
Since $\tau \hat{\omega} = \frac{\pi}{2} + 2 k \pi$, we have $\sin(\tau \hat{\omega}) = 1$ and $\cos(\tau \hat{\omega}) = 0$; substituting this, we get 
%
\begin{align*}
    \|H(j\omega;\tau)\|_{\infty} \ge \left| \frac{ \hat{\omega}^4 + \left[k_v^2 - 2 k_p \right] \hat{\omega}^2 + k_p^2 }{ \hat{\omega}^4 + \left[ \gamma^2 - 2 \gamma \hat{\omega}   \right] \hat{\omega}^2 + k_p^2} \right|.
\end{align*}
%
Now, substituting $\hat{\omega}$ as defined by~\eqref{eq:omega-ka1} into $\gamma^2 - 2 \gamma \hat{\omega}$ yields
\begin{align*}
    \gamma^2 - 2 \gamma \hat{\omega}   &= \gamma^2 - \frac{\pi\gamma}{\tau} - \frac{4 \pi \gamma k }{\tau}.
\end{align*}
In addition, substituting~\eqref{eq:k-bound} into the above, we get
\begin{align*}
    \gamma^2 - 2 \gamma \hat{\omega}   &< \gamma^2 - \frac{\pi\gamma}{\tau} - \frac{4 \pi \gamma }{\tau} \left( \frac{ ( \gamma^2 - k_v^2 + 2 k_p ) \tau }{4 \pi \gamma} - \frac{1}{4} \right) \\
    &= k_v^2 - 2k_p.
\end{align*}
Thus, if $k_a = 1$, there exists an $\omega$ such that $\| H(j \omega;\tau) \|_{\infty} > 1$.
On the basis of the discussions in the above two cases, in order to ensure $\Vert H(j \omega; \tau) \Vert_{\infty} \le 1$, we require $k_a < 1$. Therefore, the proof is completed. 
 \end{proof}

\begin{theorem} \label{theorem:string-stability-condition}
 Given any $k_a \in [0, 1)$, the vehicle platoon governed by the inter-vehicular spacing error propagation equation~\eqref{eq:inter-vehicular-spacing-error-propagation-equation-cacc-case} can be made robustly string stable for all $\tau \in (0, \tau_0]$ with feasible control gains $k_v$ and $k_p$, if the time headway satisfies 
  \begin{align} \label{eq:hw-lower-bound-theorem-3.2}
    h_w > \frac{2 \tau_0}{1 + k_a}.
  \end{align}
\end{theorem}

\begin{proof}
 First, $\vert H(j \omega;\tau) \vert^2 \le 1$ is equivalent to 
\begin{align}
 & k_a^2 \omega^4 + ( k_v^2 - 2 k_a k_p ) \omega^2 + k_p^2 \nonumber \\
&\quad \le \omega^4 - 2 \gamma \omega^3 \sin(\tau \omega) + \gamma^2 \omega^2 - 2 k_p \omega^2 \cos(\tau \omega) + k_p^2. \nonumber
\end{align}
Since $\sin(\tau \omega) \le \tau \omega$ and $\cos(\tau \omega) \le 1$, $\forall \omega \ge 0$, the above inequality is satisfied if
\begin{align}\label{eq:hinf_cond1}
 ( 1 - k_a^2 - 2 \gamma \tau ) \omega^2 + \gamma^2 - 2 k_p + 2 k_a k_p - k_v^2 \ge 0. 
 \end{align}
The inequality~\eqref{eq:hinf_cond1} is satisfied if 
 \begin{subnumcases}{} 
 1 - k_a^2 - 2 \gamma \tau \ge 0, \label{eq:suff-cond-subcase1}\\
 \gamma^2 - 2 k_p + 2 k_a k_p - k_v^2 \ge 0. \label{eq:suff-cond-subcase2}
 \end{subnumcases}
Inequalities \eqref{eq:suff-cond-subcase1} and \eqref{eq:suff-cond-subcase2} hold for all $\tau \in (0, \tau_0]$ iff
 \begin{subnumcases}{}
 \gamma \le \frac{1 - k_a^2}{2 \tau_0}, \label{eq:quadratic-inequality-condition-1-cacc} \\
  \gamma \ge \sqrt{ 2 k_p (1 - k_a) + k_v^2 }.  \label{eq:quadratic-inequality-condition-2-cacc} 
 \end{subnumcases}
Substituting $\gamma = k_v + h_w k_p$ into~\eqref{eq:quadratic-inequality-condition-1-cacc}, we have
 \begin{align}
 k_v + h_w k_p \le \frac{1 - k_a^2}{2 \tau_0},
 \end{align}
 which gives the first admissible set of $k_v$ and $k_p$ as 
 \begin{align} \label{eq:admissible-region-kv-kp-1-cacc-case}
 \mathcal{S}_{1}(k_v, k_p) \coloneqq \left\{ k_v > 0, k_p > 0, \frac{k_v}{a_1} + \frac{k_p}{b_1} \le 1 \right\},
 \end{align}
where $$a_1 = \frac{1 - k_a^2}{2 \tau_0}, \ b_1 = \frac{a_1}{h_w}.$$ 
Substituting $\gamma = k_v + h_w k_p$ into~\eqref{eq:quadratic-inequality-condition-2-cacc} and simplifying, we obtain 
 \begin{align}
 2 k_v h_w + h_w^2 k_p \ge 2 (1 - k_a), 
 \end{align}
 from which we obtain the second admissible set of $k_v$ and $k_p$ as 
 \begin{align} \label{eq:admissible-region-kv-kp-2-cacc-case}
 \mathcal{S}_{2}(k_v, k_p) \coloneqq \left\{ k_v > 0, k_p > 0, \frac{k_v}{a_2} + \frac{k_p}{b_2} \ge 1 \right\},
 \end{align}
 where $$a_2 = \frac{1 - k_a}{h_w}, \ b_2 = \frac{2 a_2}{h_w}.$$
Then, combining~\eqref{eq:admissible-region-kv-kp-1-cacc-case} and~\eqref{eq:admissible-region-kv-kp-2-cacc-case}, the admissible region of $k_v$ and $k_p$ for the inequalities~\eqref{eq:quadratic-inequality-condition-1-cacc} and~\eqref{eq:quadratic-inequality-condition-2-cacc} is given by
\begin{align} \label{eq:admissible-region-kv-kp-cacc-case}
 \mathcal{S}(k_v, k_p) = \mathcal{S}_{1}(k_v, k_p) \cap \mathcal{S}_{2}(k_v, k_p).   
\end{align}
For feasibility, we need to ensure that $\mathcal{S}(k_v, k_p)$ is non-empty. For this, we need $a_2 < a_1$, i.e., 
\begin{align}
 \frac{1 - k_a}{h_w} < \frac{1 - k_a^2}{2 \tau_0},
\end{align}
 from which we obtain~\eqref{eq:hw-lower-bound-theorem-3.2}. Therefore, the proof is completed.  
\end{proof}

 \subsubsection{Internal stability analysis}

 \begin{theorem} \label{theorem:the-inclusion-of-internal-stability-in-string-stability-cacc-case}
 The vehicle platoon governed by the inter-vehicular spacing error propagation equation~\eqref{eq:inter-vehicular-spacing-error-propagation-equation-cacc-case} is internally stable if the robust string stability conditions on $k_a$, $h_w$, $k_v$ and $k_p$ obtained from {\rm \textbf{Theorem~\ref{theorem:string-stability-condition-ka-range}}} and {\rm \textbf{Theorem~\ref{theorem:string-stability-condition}}} are satisfied. 
 \end{theorem}

\begin{proof}
 See~{\bf Appendix~\ref{Appendix_A}}. 
\end{proof}

 \subsection{CACC+ based platooning and inter-vehicular spacing error propagation equation}
 
 For the $i$-th following vehicle in the platoon, consider the following CACC+ input:
 \begin{align}
 u_i(t) = \sum\limits_{j = 1}^r \left[ k_{a j} a_{i - j}(t) - k_{v j} ( v_i(t) - v_{i - j}(t) ) \right. \nonumber \\
\left. - k_{p j} ( x_i(t) - x_{i - j}(t) + d_j + j h_w v_i(t) ) \right], 
 \end{align}
 where $k_{aj}$, $k_{vj}$, $k_{pj}$ are control gains, $d_j = j d$, $r$ denotes the number of predecessor vehicles whose motion information is available to vehicle $i$. Then the inter-vehicular spacing error propagation equation is given by
 \begin{align} \label{eq:inter-vehicular-spacing-error-propagation-equation-cacc-plus-case}
 \hat{\delta}_i(s) = \sum\limits_{j = 1}^r H_{j}(s; \tau) \hat{\delta}_{i - j}(s),
 \end{align}
 where 
 $$H_j(s; \tau) = \frac{k_{a j} s^2 + k_{v j} s + k_{p j} }{ s^2 e^{\tau s} + \sum\limits_{j = 1}^r ( k_{v j} + j k_{p j} h_w ) s + \sum\limits_{j = 1}^r k_{p j} }.
 $$ 
 Considering identical control gains for each vehicle, $k_{a j} = k_a$, $k_{v j} = k_v$, $k_{p j} = k_p$, we have 
 $$H_j(s; \tau) = \frac{k_a s^2 + k_v s + k_p}{ s^2 e^{\tau s} + \left( r k_v + \frac{ r ( r + 1 ) }{2} k_p h_w \right) s + r k_p }.
 $$

 \subsubsection{String stability analysis}

 \begin{theorem} \label{theorem:string-stability-condition-ka-range-cacc-plus-case}
 For the vehicle platoon described by the inter-vehicular spacing error propagation equation~\eqref{eq:inter-vehicular-spacing-error-propagation-equation-cacc-plus-case} to be robustly string stable, we require 
 \begin{align}
 k_a < \frac{1}{r}. 
 \end{align} 
\end{theorem}

\begin{proof}
 First, according to~\cite{10.1115/1.2802497}, the string stability of~\eqref{eq:inter-vehicular-spacing-error-propagation-equation-cacc-plus-case} can be ensured if 
 \begin{align}
 \Vert r H_j(s; \tau) \Vert_{\infty} &= \left\Vert \frac{ r k_a s^2 + r k_v s + r k_p }{ s^2 e^{\tau s} + ( r k_v + \frac{ r ( 1 + r ) }{2} k_p h_w ) s + r k_p } \right\Vert_{\infty} \nonumber \\
 &\le 1. 
 \end{align}
Define $\tilde{k}_a := r k_a$, $\tilde{k}_v := r k_v$, $\tilde{k}_p := r k_p$, $\tilde{h}_w := \frac{ 1 + r }{2} h_w$, then the above inequality can be rewritten as 
 \begin{align} \label{eq:string-stability-condition-cacc-plus-case-with-identical-control-gains}
 \left\Vert \frac{ \tilde{k}_a s^2 + \tilde{k}_v s + \tilde{k}_p }{ s^2 e^{\tau s} + ( \tilde{k}_v + \tilde{k}_p \tilde{h}_w ) s + \tilde{k}_p } \right\Vert_{\infty} \le 1. 
 \end{align}
Thus, based on~{\rm \textbf{Theorem~\ref{theorem:string-stability-condition-ka-range}}}, we have
$\tilde{k}_a < 1$, that is, $r k_a < 1$, or $k_a < 1/r$. Thus, the proof is completed. 
\end{proof}

\begin{theorem} \label{theorem:string-stability-condition-cacc-plus-case}
 Given any $k_a \in [0, 1/r)$, the vehicle platoon described by~\eqref{eq:inter-vehicular-spacing-error-propagation-equation-cacc-plus-case} can be made robustly string stable for all $\tau \in (0, \tau_0]$ with feasible control gains $k_v$ and $k_p$, if the time headway satisfies 
 \begin{align}\label{eq:time-headway-cacc+}
 h_w > \frac{ 4 \tau_0 }{ ( 1 + r ) ( 1 + r k_a ) }.
 \end{align}
\end{theorem}

\begin{proof}
Applying the result in~{\rm \textbf{Theorem~\ref{theorem:string-stability-condition}}} to \eqref{eq:string-stability-condition-cacc-plus-case-with-identical-control-gains}, we obtain 
\begin{align} \label{eq:range-of-tilde-hw}
 \tilde{h}_w > \frac{2 \tau_0}{ 1 + \tilde{k}_a }.
\end{align}
Substituting $\tilde{h}_w = \frac{1+r}{2} h_w$ and $\tilde{k}_a = r k_a$ and simplifying, we obtain \eqref{eq:time-headway-cacc+}.

In addition, the admissible region of the control gains $k_v$ and $k_p$ is given by
 \begin{align} \label{eq:kv-kp-admissible-region-cacc+}
 \tilde{\mathcal{S}}(k_v, k_p) = \tilde{\mathcal{S}}_1(k_v, k_p) \cap \tilde{\mathcal{S}}_2(k_v, k_p),
 \end{align}
 where $\tilde{\mathcal{S}}_1(k_v, k_p)$ and $\tilde{\mathcal{S}}_2(k_v, k_p)$ are respectively defined as 
 \begin{align}
 \tilde{\mathcal{S}}_1(k_v, k_p) \coloneqq \left\{ k_v > 0, k_p > 0, \frac{k_v}{\tilde{a}_1} + \frac{k_p}{\tilde{b}_1} \le \frac{1}{r} \right\},
 \end{align}
 \begin{align}
 \tilde{\mathcal{S}}_2(k_v, k_p) \coloneqq \left\{ k_v > 0, k_p > 0, \frac{k_v}{\tilde{a}_2} + \frac{k_p}{\tilde{b}_2} \ge \frac{1}{r} \right\},
 \end{align}
 in which, $\tilde{a}_1$, $\tilde{b}_1$ and $\tilde{a}_2$, $\tilde{b}_2$ are respectively given by
 \begin{align}
 \tilde{a}_1 = \frac{ 1 - r^2 k_a^2 }{ 2 \tau_0 }, \ \tilde{b}_1 = \frac{\tilde{a}_1}{\tilde{h}_w}, \ \tilde{a}_2 = \frac{ 1 - r k_a }{ \tilde{h}_w }, \ \tilde{b}_2 = \frac{2 \tilde{a}_2}{\tilde{h}_w}.
\end{align}  
 Based on the result for the CACC case, if $\tilde{h}_w$ satisfies~\eqref{eq:range-of-tilde-hw}, i.e. $h_w$ satisfies~\eqref{eq:time-headway-cacc+}, then the admissible region of $k_v$ and $k_p$ given by~\eqref{eq:kv-kp-admissible-region-cacc+} is non-empty. Therefore, the proof is completed. 
\end{proof}

 \subsubsection{Internal stability analysis}

 \begin{theorem} \label{theorem:the-inclusion-of-internal-stability-in-string-stability-cacc-plus-case} 
 The robust string stability conditions derived in~{\rm \textbf{Theorem~\ref{theorem:string-stability-condition-ka-range-cacc-plus-case}}} and~{\rm \textbf{Theorem~\ref{theorem:string-stability-condition-cacc-plus-case}}} can ensure internal stability of the vehicle platoon described by the inter-vehicular spacing error propagation equation~\eqref{eq:inter-vehicular-spacing-error-propagation-equation-cacc-plus-case}.  
\end{theorem}

\begin{proof}
  See~{\bf Appendix~\ref{Appendix_B}}.   
\end{proof}
     \section{Simulation Results and Discussion} \label{section:numerical-simulations}
 In this section, we present a numerical example to evaluate the theoretical results in Section~\ref{section:main-results}. We consider the following numerical values for the system parameters: $N = 10$, $\tau_0 = 0.5$ s, $d = 5$ m, respectively. In the numerical simulation, $\tau$ was chosen as $\tau = \tau_0$. The initial steady-state velocity for the platoon is assumed to be $25$ m/s. The perturbation on the acceleration of the lead vehicle is assumed to be
\begin{align}
 a_0(t) = 
 \begin{cases}
 0.5 \sin ( 0.1 \pi (t - 10) ), t \in (10, 30)~{\rm s}, \\
  0, \mbox{otherwise}.
  \end{cases}
\end{align}
Under the above chosen values of system parameters and acceleration perturbation model of the lead vehicle, we will evaluate the performance of the platoon under CACC and CACC+ in the following two subsections, respectively.
\subsection{CACC case}
First, following from~{\rm \textbf{Theorem~\ref{theorem:string-stability-condition-ka-range}}}, choose $k_a = 0.5$, then according to~{\rm \textbf{Theorem~\ref{theorem:string-stability-condition}}}, the lower bound of $h_w$ can be computed as 
\begin{align} \label{eq:lower-bound-of-hw-in-simulation}
 h_w > \frac{2 \tau_0}{1 + k_a} = 0.6667~{\rm s}.
\end{align} 
Then, choosing $h_w = 0.7$ s, we have $a_1 = 0.7500$, $b_1 = 1.0714$, $a_2 = 0.7143$, $b_2 = 2.0408$, and accordingly the feasible region of $k_v$ and $k_p$ can be obtained as shown in Figure~\ref{fig:kvkp-CACC-case}.
 \begin{figure}[!htb]
\centering{\includegraphics[scale=0.6]{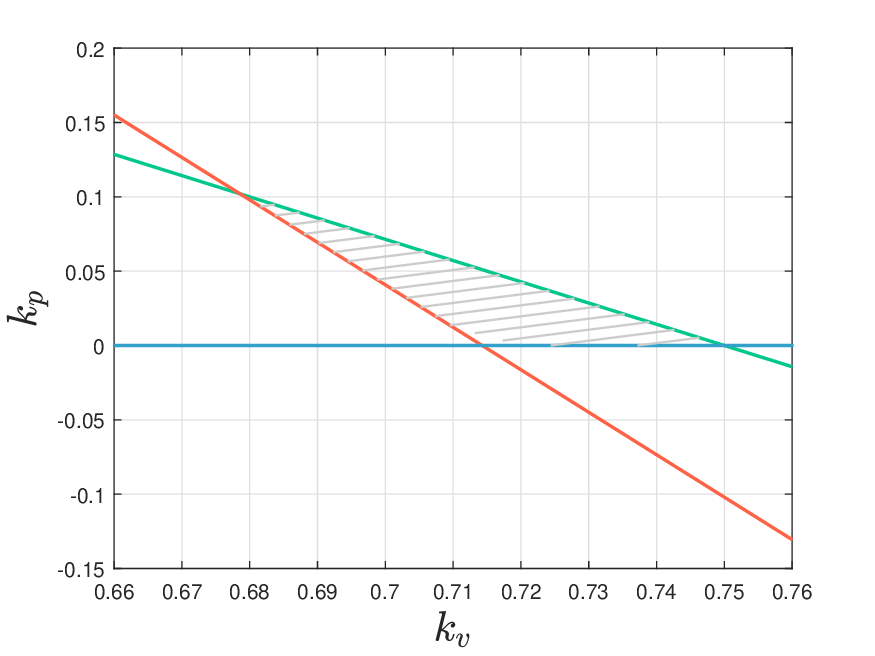}}
\caption{The admissible region of $k_v$ and $k_p$ (CACC case).
\label{fig:kvkp-CACC-case}}
\end{figure}
Choose $k_v = 0.7$, $k_p = 0.06$, 
then the relation between $\vert H(j \omega; \tau) \vert$ and $\tau$ and $\omega$ is shown in Figure~\ref{fig:H_norm}. It can be seen that the robust string stability condition $\vert H(j \omega; \tau) \vert \le 1$, $\forall \tau \in (0, \tau_0]$ and $\forall \omega \ge 0$ is satisfied under the chosen time headway and control gains.
\begin{figure}[!htb]
\centering{\includegraphics[scale=0.6]{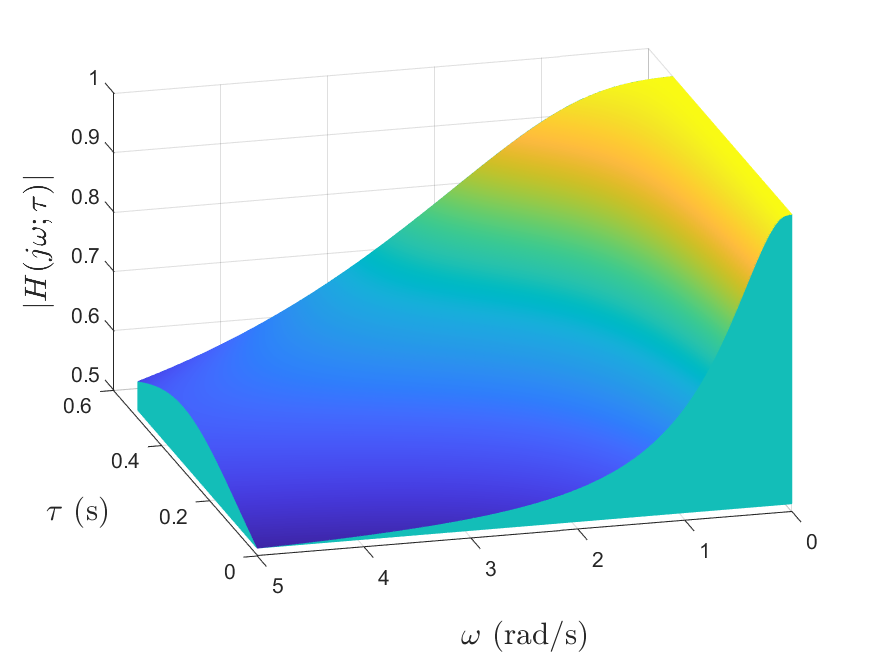}}
\caption{The profile of $\vert H(j \omega; \tau) \vert$ (CACC case).
\label{fig:H_norm}}
\end{figure}

To verify the internal stability of the closed-loop system, first, the responses and the roots of the real and imaginary parts of $\tau_0^2 \mathcal{D}(j \theta)$ are presented in Figure~\ref{fig:internal_stability}, from which it can be seen that (i) the equations $\mathcal{D}_i(\theta) = 0$ and $\mathcal{D}_r(\theta) = 0$ have only real roots; (ii) these roots interlace. 
\begin{figure}[!htb]
\centering{\includegraphics[scale=0.6]{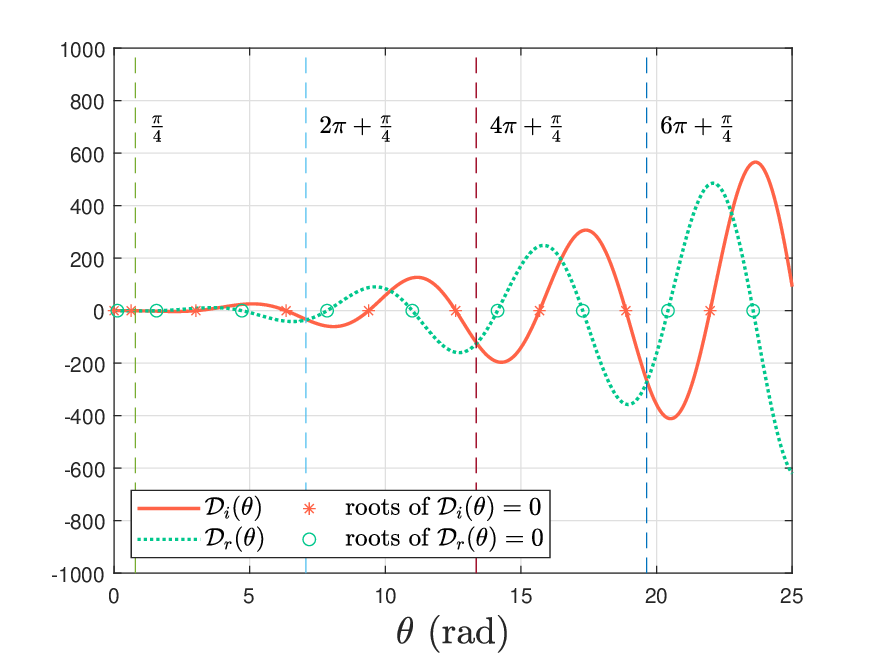}}
\caption{The responses and roots of the real and imaginary parts of $\tau_0^2 \mathcal{D}(j \theta)$ where $\theta = \tau_0 \omega$.
\label{fig:internal_stability}}
\end{figure}
In addition, substituting the values of the time headway and the control gains as well as $\tau = \tau_0$ into~\eqref{eq:internal-stability-condition-2}, the plot of $\omega$ vs.  $\mathcal{D}_i^{\prime}(\omega) \mathcal{D}_r(\omega) - \mathcal{D}_i(\omega) \mathcal{D}_r^{\prime}(\omega)$ is shown in Figure~\ref{fig:internal_stability_condition_2}. 
\begin{figure}[!htb]
\centering{\includegraphics[scale=0.6]{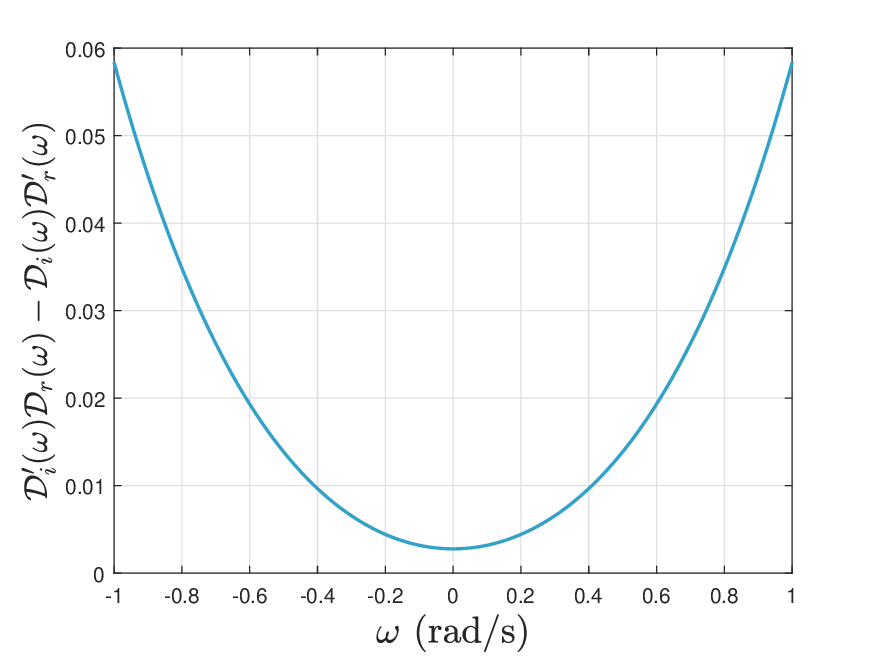}}
\caption{The curve of $\mathcal{D}_i^{\prime}(\omega) \mathcal{D}_r(\omega) - \mathcal{D}_i(\omega) \mathcal{D}_r^{\prime}(\omega)$.
\label{fig:internal_stability_condition_2}}
\end{figure}
Note that $\mathcal{D}_i^{\prime}(0) \mathcal{D}_r(0) - \mathcal{D}_i(0) \mathcal{D}_r^{\prime}(0) > 0$ and $\mathcal{D}_i^{\prime}(\omega) \mathcal{D}_r(\omega) - \mathcal{D}_i(\omega) \mathcal{D}_r^{\prime}(\omega)$ is positive for all $\omega$ under the chosen time headway and control gains. Therefore, the closed-loop system is internally stable. 

Under the above chosen values of the time headway $h_w$ and the control gains $k_a$, $k_v$, $k_p$, the response of the inter-vehicular spacing errors is provided in Figure~\ref{fig:delta_CACC}.
\begin{figure}[!htb]
\centering{\includegraphics[scale=0.6]{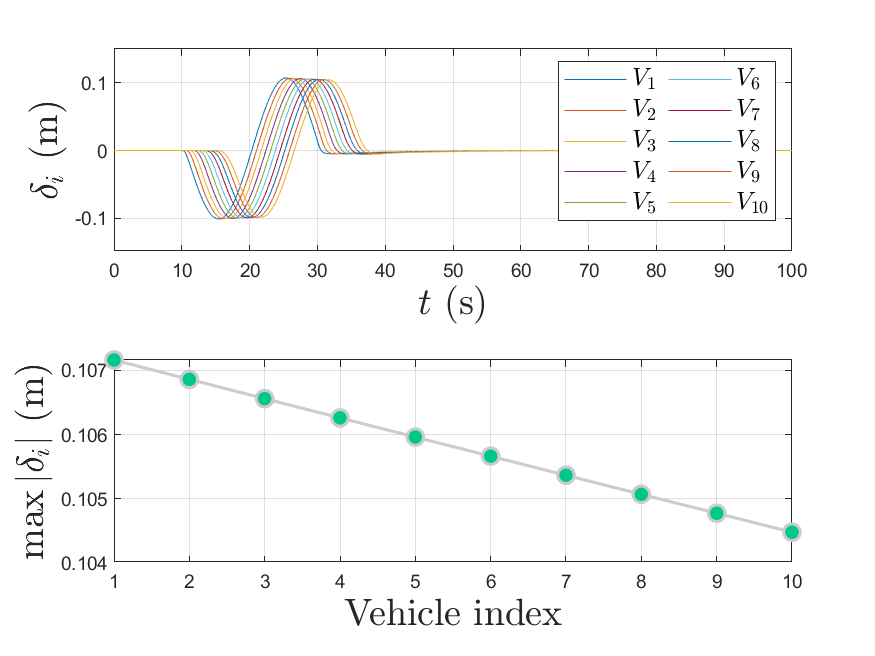}}
\caption{The inter-vehicular spacing errors (CACC case, $h_w = 0.7$ s).
\label{fig:delta_CACC}}
\end{figure}
In contrast, when $h_w = 0.6$ s, the evolution of the inter-vehicular spacing errors is given in Figure~\ref{fig:delta_CACC_hw_0point6}, from which it can be seen that the platoon is string unstable when $h_w = 0.6$ s, that is smaller than the lower bound of the time headway in~\eqref{eq:lower-bound-of-hw-in-simulation}.  
\begin{figure}[!htb]
\centering{\includegraphics[scale=0.6]{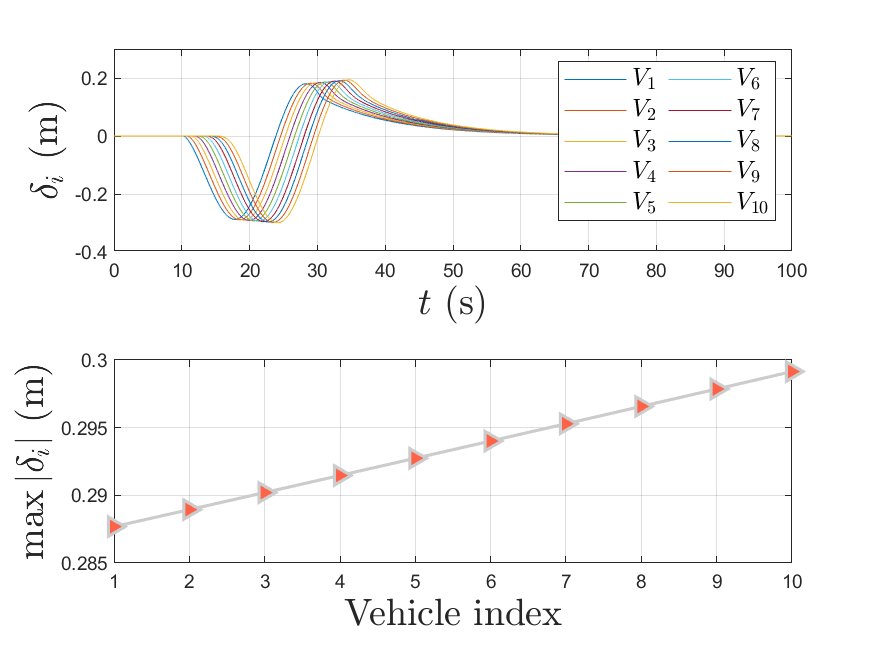}}
\caption{The inter-vehicular spacing errors (CACC case, $h_w = 0.6$ s).
\label{fig:delta_CACC_hw_0point6}}
\end{figure}

 When $k_a = 0$ in the CACC case, we obtain the ACC case, and the lower bound of $h_w$ for ensuring robust string stability becomes $h_w > 2 \tau_0 = 1$ s, then choosing $h_w = 1.2$ s, we can compute
$a_1 = 1, b_1 = 0.8333, a_2 = 0.8333, b_2 = 1.3889$, and accordingly the admissible region of $k_v$ and $k_p$ is shown in Figure~\ref{fig:kvkp-ACC-case}.
\begin{figure}[!htb]
\centering{\includegraphics[scale=0.6]{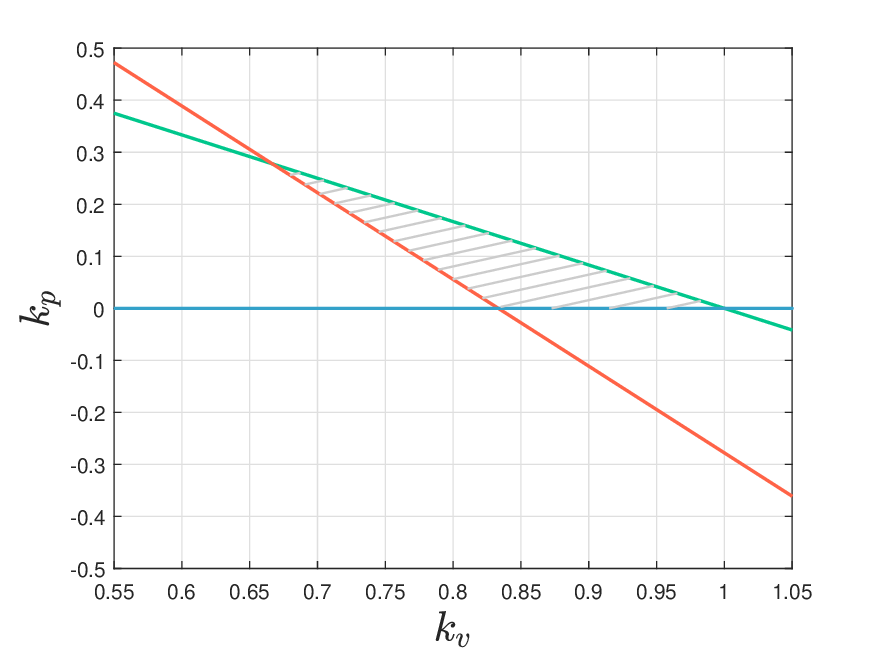}}
\caption{The admissible region of $k_v$ and $k_p$ (ACC case).
\label{fig:kvkp-ACC-case}}
\end{figure}
Choosing $k_v = 0.8$, $k_p = 0.1$ from the admissible region, the relation between $\vert H(j \omega; \tau) \vert$ and $\tau$ and $\omega$ is shown in Figure~\ref{fig:H_1_norm}. 
\begin{figure}[!htb]
\centering{\includegraphics[scale=0.6]{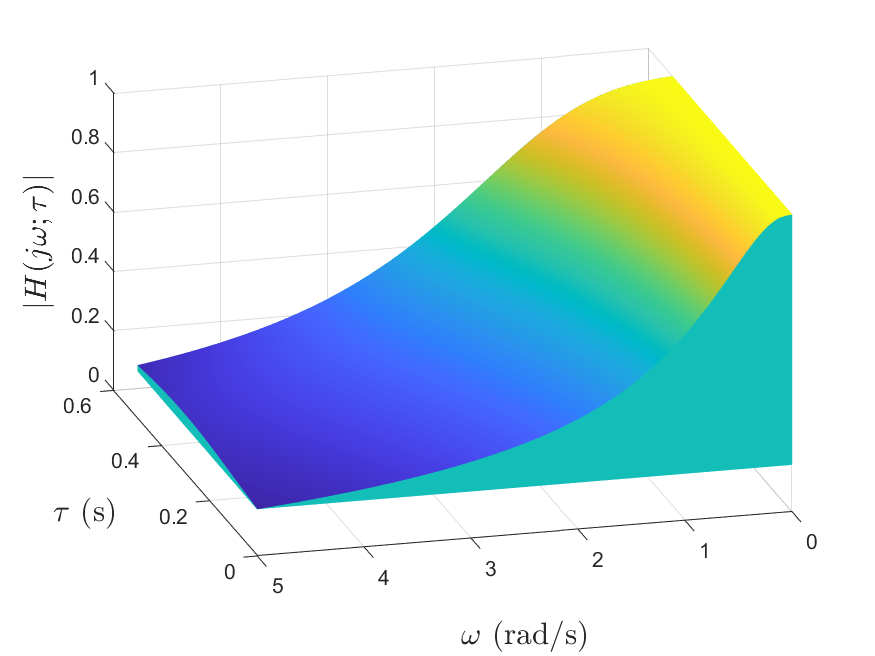}}
\caption{The profile of $\vert H(j \omega; \tau) \vert$ (ACC case).
\label{fig:H_1_norm}}
\end{figure}
Under the above chosen values of the time headway $h_w$ and the control gains $k_v$, $k_p$, the evolution of the inter-vehicular spacing errors is given in Figure~\ref{fig:delta_ACC}.  
\begin{figure}[!htb]
\centering{\includegraphics[scale=0.6]{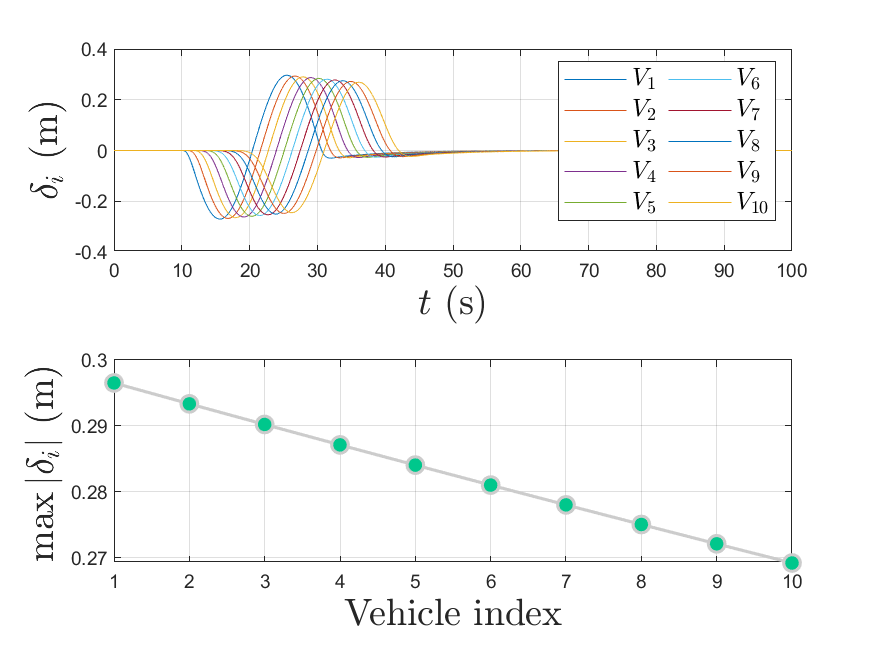}}
\caption{The inter-vehicular spacing errors (ACC case, $h_w = 1.2$ s).
\label{fig:delta_ACC}}
\end{figure}
For comparison, when $h_w$ is changed from 1.2 s to 0.9 s, the response of the inter-vehicular spacing errors for ACC is shown in Figure~\ref{fig:delta_ACC_hw_0point9}, which exhibits string instability.  
\begin{figure}[!htb]
\centering{\includegraphics[scale=0.6]{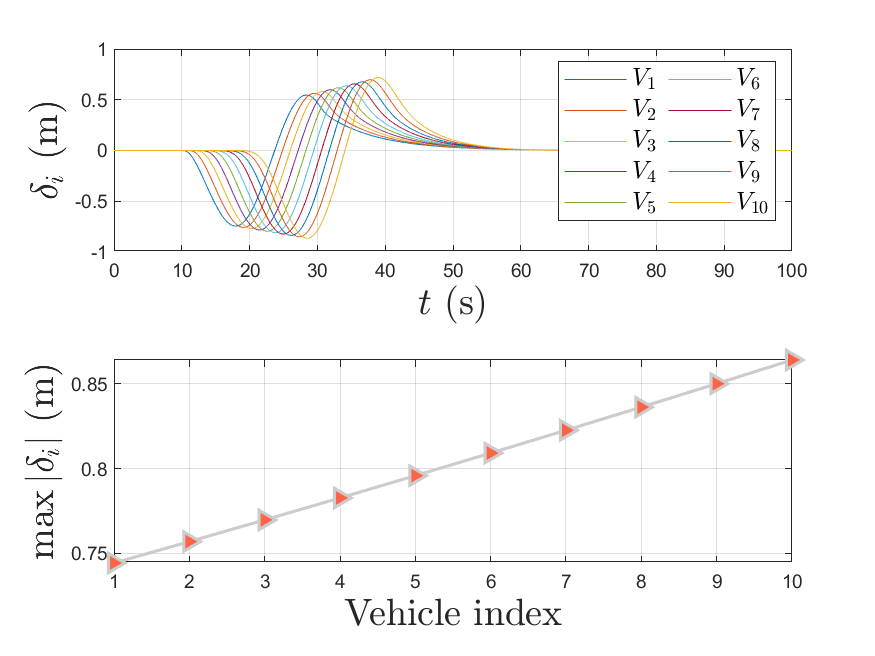}}
\caption{The inter-vehicular spacing errors (ACC case, $h_w = 0.9$ s).
\label{fig:delta_ACC_hw_0point9}}
\end{figure}

\subsection{CACC+ case}
Assume $r = 3$. First, according to~{\rm \textbf{Theorem~\ref{theorem:string-stability-condition-ka-range-cacc-plus-case}}}, we obtain $k_a < \frac{1}{3}$, from which we choose $k_a = 0.2$. In addition, based on~{\rm \textbf{Theorem~\ref{theorem:string-stability-condition-cacc-plus-case}}}, we have
\begin{align}
 h_w > \frac{ 4 \tau_0 }{ ( 1 + r ) ( 1 + r k_a ) } = 0.3125~{\rm s}.
\end{align} 
Choosing $h_w = 0.32$, we can compute 
 \begin{align}
 \tilde{a}_1 = 0.64, \tilde{b}_1 = 1, \tilde{a}_2 = 0.6250, \tilde{b}_2 = 1.9531, 
 \end{align} 
 accordingly, the admissible region of $k_v$ and $k_p$ is shown in Figure~\ref{fig:kvkp-CACC-plus-case}. Choose $k_v = 0.206$, $k_p = 0.01$, 
 \begin{figure}[!htb]
\centering{\includegraphics[scale=0.6]{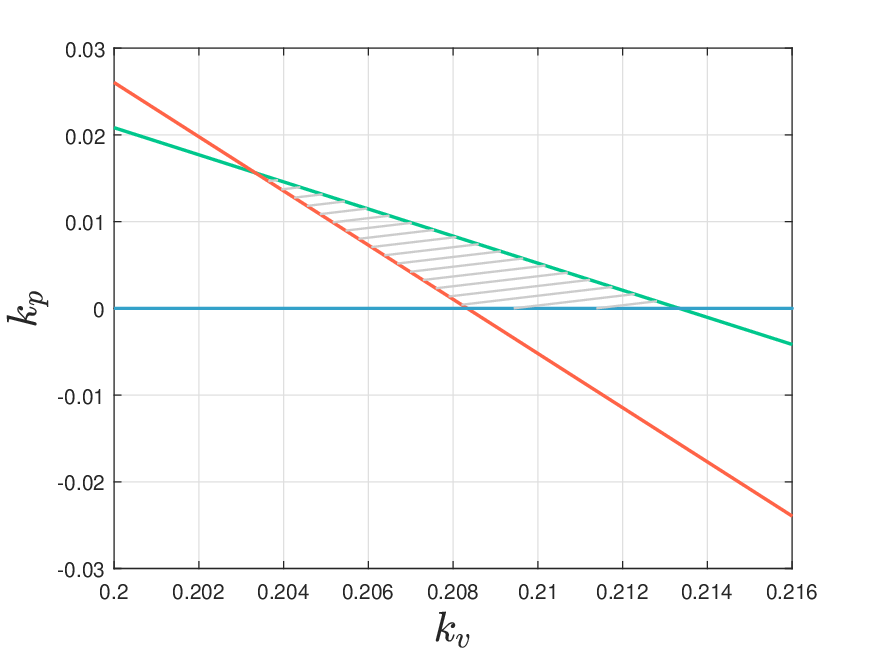}}
\caption{The admissible region of $k_v$ and $k_p$ (CACC+ case).
\label{fig:kvkp-CACC-plus-case}}
\end{figure}
then the norm of $r H(j \omega; \tau)$ is shown in Figure~\ref{fig:H_norm_cacc_plus}.
\begin{figure}[!htb]
\centering{\includegraphics[scale=0.6]{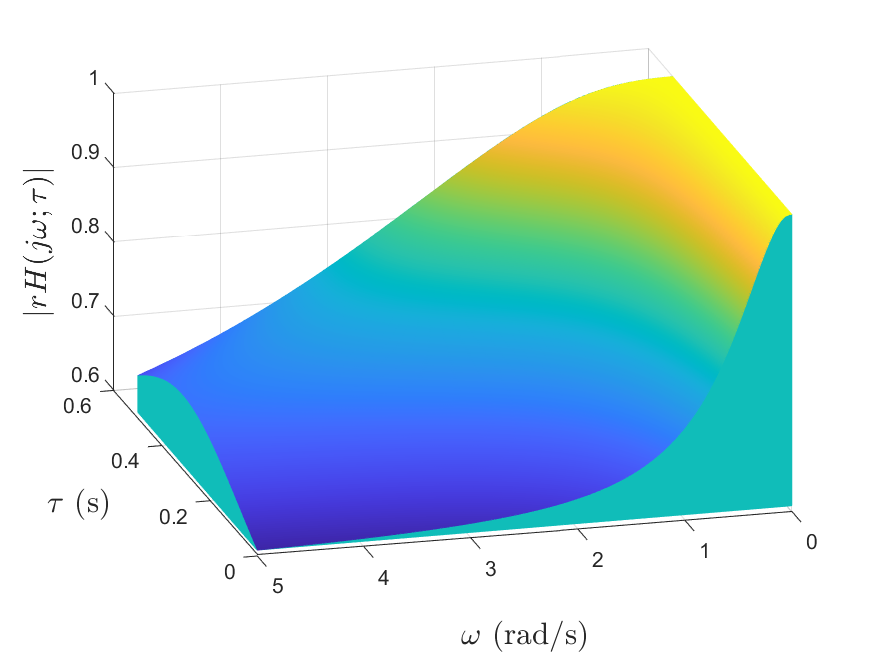}}
\caption{The profile of $\vert r H(j \omega; \tau) \vert$ (CACC+ case).
\label{fig:H_norm_cacc_plus}}
\end{figure}

Note that for the first two following vehicles, i.e. $V_1$ and $V_2$ as illustrated in Figure~\ref{fig_illustration_of_vehicle_platoon}, there are less than 3 predecessor vehicles, thus we choose $r = 2$ for $V_2$ and $r = 1$ for $V_1$. Through the same procedure for $r = 3$, we can obtain a set of feasible time headway and control gains for $V_2$ as $h_w = 0.5$, $k_a = 0.2$, $k_v = 0.4$, $k_p = 0.02$; the time headway and control gains for $V_1$ are chosen the same as those adopted in the CACC case. Under the above chosen values of the time headway and the control gains, the response of the inter-vehicular spacing errors is given by Figure~\ref{fig:delta_CACC_plus}, from which it can be observed that the inter-vehicular spacing errors are not amplified from the entire platoon point of view.          
\begin{figure}[!htb]
\centering{\includegraphics[scale=0.6]{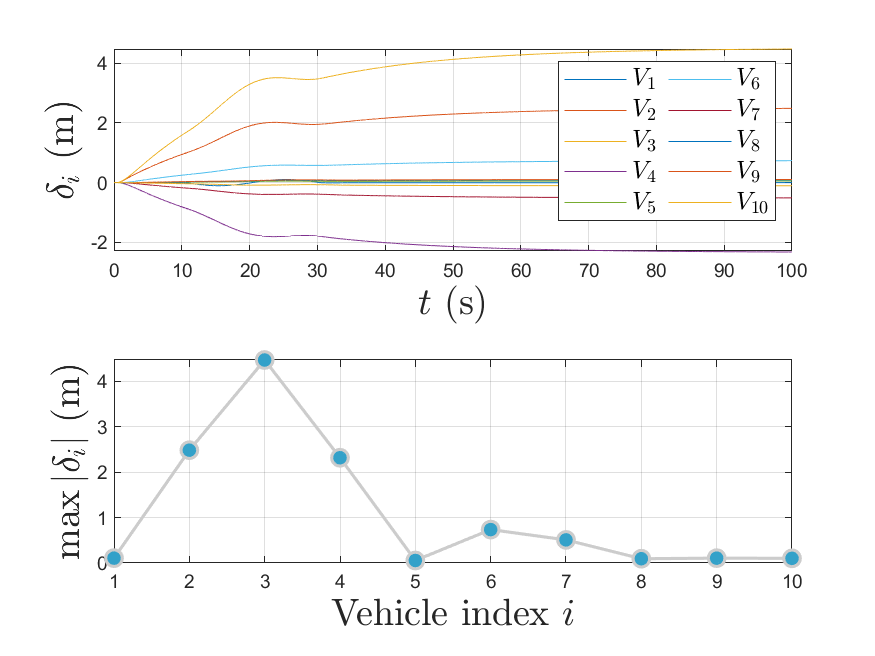}}
\caption{The inter-vehicular spacing errors (CACC+ case).
\label{fig:delta_CACC_plus}}
\end{figure}

From the simulation results, it can be seen that the vehicle platoon can achieve both string stability and internal stability for ACC, CACC and CACC+ cases with the time headway and control gains synthesized through the design procedure developed in Section~\ref{section:main-results}, which corroborates the main results.
     \section{Conclusion} \label{section:conclusion}
 We have proved that a connected and autonomous vehicle platoon under adaptive cruise control and cooperative adaptive cruise control systems is robustly string stable with general parasitic actuation delay, that is, given an upper bound on the parasitic actuation delay ($\tau_0$), one can select time headways satisfying $h_w > 2 \tau_0$ for ACC, $h_w > \frac{2 \tau_0}{1 + k_a}$ for CACC and $h_w > \frac{4 \tau_0}{(1 + r)(1 + r k_a)}$ for CACC+ with information from $r$ predecessors. This work answers the long-standing question in the affirmative as to whether one can select the time headway for the platoon and ensure robust string stability in the presence of general parasitic actuation delay. We have shown that such a selection of the time headway also ensures internal stability of the platoon. Numerical simulations on an illustrative example corroborate this result. Future work will consider the effect of communicated information packet drops on the lower bound of the time headway under general parasitic actuation delay.


 \appendix
 

\section{} \label{Appendix_A}
 \begin{proof}[Proof of Theorem~\ref{theorem:the-inclusion-of-internal-stability-in-string-stability-cacc-case}]
From~\eqref{eq:inter-vehicular-spacing-error-propagation-equation-cacc-case}, we have 
\begin{align} \label{eq:D-ast-s}
\tau^2 \mathcal{D}(s) = \tau^2 (s^2 e^{\tau s} + \gamma s + k_p).
\end{align}
Then, substituting $s = j \omega$, we have
 \begin{align} \label{eq:D-ast-j-omega}
 \tau^2 \mathcal{D}(j \omega) = \mathcal{D}_r(\omega) + j \mathcal{D}_i(\omega),
 \end{align}
 where $\mathcal{D}_r(\omega) = \tau^2 k_p - \tau^2 \omega^2 \cos(\tau \omega)$, $\mathcal{D}_i(\omega) = \tau^2 \gamma \omega - \tau^2 \omega^2 \sin(\tau \omega)$. Let $\theta = \tau \omega, \bar{\gamma} = \tau \gamma, \bar{k}_p = \tau^2 k_p$, then $\mathcal{D}_r(\omega)$ and $\mathcal{D}_i(\omega)$ can be respectively rewritten as follows:
\[
\mathcal{D}_r(\theta) = \bar{k}_p - \theta^2 \cos\theta; \ \  \mathcal{D}_i(\theta) = \bar{\gamma}\theta - \theta^2 \sin\theta.
\]
Since $\mathcal{D}_r(0) \neq 0$, roots of $\mathcal{D}_r(\theta) = 0$ are given by
\begin{align} \label{eq:cosine-real-part-equation}
\widetilde{\mathcal{D}}_r(\theta) \coloneqq \frac{\bar{k}_p}{\theta^2} - \cos\theta = 0.
\end{align}
The roots of $\mathcal{D}_i(\theta) = 0$ are $\theta=0$ and the roots of 
\begin{align} \label{eq:sine-imaginary-part-equation}
 \widetilde{\mathcal{D}}_i(\theta) \coloneqq \frac{\bar{\gamma}}{\theta} - \sin\theta = 0.
\end{align}
Denote the non-negative roots of $\mathcal{D}_i(\theta) = 0$ as $0 = \theta_{i, 1} < \theta_{i, 2} < \cdots$, and the positive roots of $\mathcal{D}_r(\theta) = 0$ as $\theta_{r, 1} < \theta_{r, 2} < \cdots$. To prove internal stability, according to~{\rm \textbf{Lemma~\ref{lemma:internal-stability}}}, we need to prove the following:
\begin{enumerate}[(A)]
    \item there are only simple and real roots of $\mathcal{D}_i(\theta)=0$ and $\mathcal{D}_r(\theta)=0$, and these roots interlace, i.e., $\theta_{i, 1} < \theta_{r, 1} < \theta_{i, 2} < \theta_{r, 2} < \cdots$; 
    \item $\mathcal{D}_i^{\prime}(\omega_0) \mathcal{D}_r(\omega_0) - \mathcal{D}_i(\omega_0) \mathcal{D}_r^{\prime}(\omega_0) > 0$ for some $\omega_0 \in (-\infty, \infty)$. \\
\end{enumerate}

\noindent \textbf{Proof of Statement (A)} 

According to {\rm \textbf{Lemma~\ref{lemma:real-roots}}}, we will employ the following procedure to show statement (A). First, we will separately show the roots of $\mathcal{D}_r(\theta) = 0$ or Equation \eqref{eq:cosine-real-part-equation} and those of $\mathcal{D}_i(\theta) = 0$ or Equation \eqref{eq:sine-imaginary-part-equation} lie in distinct intervals and they are simple and real. Then, we will show the interlacing property. \\

\noindent \textbf{Roots of $\mathcal{D}_r(\theta) = 0$ or Equation \eqref{eq:cosine-real-part-equation}} \\

\textbf{Step 1:} We will first show that $\bar k_p < \frac{4}{27} 
$.

Combining~\eqref{eq:quadratic-inequality-condition-1-cacc} and~\eqref{eq:quadratic-inequality-condition-2-cacc}, we have  
\begin{align}
 \sqrt{ 2 k_p (1 - k_a) + k_v^2 } \le \frac{1 - k_a^2}{2 \tau_0},
\end{align}
 from which we obtain 
\begin{align}
 k_p \le \frac{ \left( \frac{1 - k_a^2}{2 \tau_0} \right)^2 - k_v^2 }{ 2 (1 - k_a) }.
\end{align}
Then, it follows that
\begin{align} \label{eq:tau-square-kp-inequality-1}
 \bar{k}_p = \tau^2 k_p \le \tau_0^2 k_p < \frac{ \tau_0^2 \left( \frac{ 1 - k_a^2 }{2 \tau_0} \right)^2 }{2 (1 - k_a)} = \frac{ (1 - k_a) (1 + k_a)^2 }{8}.
\end{align}
Define
\begin{align}
 f_1(k_a) &\coloneqq (1 - k_a) (1 + k_a)^2 \nonumber \\
&= - k_a^3 - k_a^2 + k_a + 1.  
\end{align}
Then, 
 \begin{align}
 \frac{d f_1(k_a)}{d k_a} 
 = (1 - 3 k_a) ( 1 + k_a ).
 \end{align}
Thus, when $k_a \in [0, 1)$, we have
\begin{align}
 \max\limits_{k_a \in [0, 1)} f_1(k_a) = f_1\left(\frac{1}{3}\right) = \frac{32}{27}.  
\end{align}
Substituting this into~\eqref{eq:tau-square-kp-inequality-1} yields
\begin{align} \label{eq:bar-kp-upper-bound}
 \bar{k}_p < \frac{\max\limits_{k_a \in [0, 1)} f_1(k_a)}{8} = \frac{4}{27}.
\end{align}

\textbf{Step 2:} We will show that there is exactly one root of $\mathcal{D}_r(\theta)$ in the interval $\left(0, \frac{\pi}{4} \right)$. 

It suffices to show that $\mathcal{D}_r(\theta)$ is monotonic in the interval $\left(0, \frac{\pi}{4} \right)$ with $\mathcal{D}_r(0)\mathcal{D}_r\left(\frac{\pi}{4} \right)<0$.

Notice that \begin{eqnarray*}-\frac{d}{d\theta}\mathcal{D}_r(\theta) &=& 2 \theta \cos \theta - \theta ^2 \sin \theta 
= \theta (2 \cos \theta - \theta \sin \theta)\\ &>& \theta \left(\sqrt{2} - \frac{\sqrt{2}\pi}{8} \right)>0.
\end{eqnarray*}
Hence, $\mathcal{D}_r(\theta)$ is monotonically decreasing in the interval $\left(0, \frac{\pi}{4} \right)$; moreover 
\begin{eqnarray*}
   \mathcal{D}_r(0) &=& {\bar k}_p >0, \\ \mathcal{D}_r\left(\frac{\pi}{4} \right) &=& {\bar k}_p - \left(\frac{\pi}{4} \right)^2 \frac{\sqrt{2}}{2} <\frac{4}{27}- \frac{\sqrt{2}\pi^2}{32} <0. 
\end{eqnarray*}
Therefore, there exists only one real root $\theta_{r,1} \in \left(0,\frac{\pi}{4} \right)$. \\


\textbf{Step 3:} Similarly, we can show that there exists only one real root in the interval $\left(\frac{\pi}{4}, \frac{\pi}{2} \right)$, no real root in the interval $\left(\frac{\pi}{2}, \frac{3\pi}{2} \right)$, and only one real root in the interval $\left(\frac{3 \pi}{2}, \frac{3 \pi}{2} + \frac{\pi}{4} \right)$ for $\mathcal{D}_r(\theta) = 0$. To this end, we will utilize equation~\eqref{eq:cosine-real-part-equation} since it has the same real roots as $\mathcal{D}_r(\theta) = 0$ for $\theta \neq 0$. 

When $\theta \in \left(\frac{\pi}{4}, \frac{\pi}{2} \right)$: As
 \begin{align}
 \frac{2 \bar{k}_p}{\theta^3} < \frac{ 2 \bar{k}_p }{ \left(\frac{\pi}{4} \right)^3 } < \frac{ \frac{8}{27} }{ \left( \frac{\pi}{4} \right)^3 } < \frac{ \sqrt{2} }{2} < \sin\theta,    
 \end{align}
 then 
 \begin{align}
  \frac{d}{d \theta} \widetilde{\mathcal{D}}_r(\theta) = -\frac{ 2 \bar{k}_p }{ \theta^3 } + \sin \theta > 0.   
 \end{align}
 Thus, $\widetilde{\mathcal{D}}_r(\theta)$ increases monotonically in this interval; in addition, 
it follows from~\eqref{eq:bar-kp-upper-bound} that  
 \begin{align}
 \widetilde{\mathcal{D}}_r\left(\frac{\pi}{4} \right) &= \frac{\bar{k}_p}{\left(\frac{\pi}{4} \right)^2} - \cos\left(\frac{\pi}{4} \right) < 0, \\
 \widetilde{\mathcal{D}}_r\left(\frac{\pi}{2} \right) &= \frac{\bar{k}_p}{\left(\frac{\pi}{2} \right)^2} - \cos\left(\frac{\pi}{2} \right) > 0.
 \end{align}
Therefore, there exists only one real root $\theta_{r, 2} \in \left( \frac{\pi}{4}, \frac{\pi}{2} \right)$. 

When $\theta \in \left( \frac{\pi}{2}, \frac{3 \pi}{2} \right)$: As $\cos\theta < 0 < \frac{\bar{k}_p}{\theta^2}$, then $\widetilde{\mathcal{D}}_r(\theta) > 0$ for all $\theta \in (\frac{\pi}{4}, \frac{\pi}{2})$, i.e. 
 there is no real root in this interval. 
 
When $\theta \in \left( \frac{3 \pi}{2}, \frac{3 \pi}{2} + \frac{\pi}{4} \right)$: As
 %
 \begin{align}
 - \frac{2 \bar{k}_p}{\theta^3} < 0, \quad \sin \theta < 0,    
 \end{align}
 then 
 \begin{align}
 \frac{d}{d \theta} \widetilde{\mathcal{D}}_r(\theta) = - \frac{2 \bar{k}_p}{\theta^3} + \sin \theta < 0,    
 \end{align}
 which indicates $\widetilde{\mathcal{D}}_r(\theta)$ is monotonically decreasing in this interval; also, the following are satisfied for $\theta \in \left(\frac{3 \pi}{2}, \frac{3 \pi}{2} + \frac{\pi}{4} \right)$:  
 %
 \begin{align}
 \widetilde{\mathcal{D}}_r\left(\frac{3 \pi}{2} \right) = \frac{ \bar{k}_p }{ \left( \frac{ 3 \pi}{2} \right)^2 } - \cos\left( \frac{3 \pi}{2} \right) > 0, \\
\widetilde{\mathcal{D}}_r\left(\frac{3 \pi}{2} + \frac{\pi}{4} \right) < \frac{ \bar{k}_p }{ \left( \frac{\pi}{4} \right)^2 } - \cos\left( \frac{3 \pi}{2} + \frac{\pi}{4} \right) < 0. 
 \end{align}
Thus there exists only one real root $\theta_{r, 3} \in \left( \frac{3 \pi}{2}, \frac{3 \pi}{2} + \frac{\pi}{4} \right)$.

 Repeating the above procedure and using the periodicity of $\cos\theta$, we can obtain the following: 
 
 When $\theta \in \left( \frac{\pi}{4} + 2 m \pi, \frac{\pi}{2} + 2 m \pi \right)$, where $m = 0, 1, 2, \cdots$, we have 
 \begin{align}
 \frac{2 \bar{k}_p}{\theta^3} < \frac{2 \bar{k}_p}{\left(\frac{\pi}{4} \right)^2} < \frac{\sqrt{2}}{2} < \sin\theta,
 \end{align}
%
 i.e., 
 \begin{align}
 \frac{d}{d \theta} \widetilde{\mathcal{D}}_r(\theta) = - \frac{2 \bar{k}_p}{\theta^3} + \sin \theta > 0,    
 \end{align}
 as well as 
 \begin{align}
 \begin{cases}{}
 \cos( \frac{\pi}{4} + 2 m \pi) > \frac{ \bar{k}_p }{ \left(\frac{\pi}{4} + 2 m \pi \right)^2 }, \\
 \cos( \frac{\pi}{2} + 2 m \pi ) < \frac{ \bar{k}_p }{ \left(\frac{\pi}{2} + 2 m \pi \right)^2 }, 
 \end{cases}
 \end{align}
 i.e., 
 \begin{align}
 \begin{cases}
\widetilde{\mathcal{D}}_r\left( 
\frac{\pi}{4} + 2 m \pi \right) = \frac{\bar{k}_p}{ \left( \frac{\pi}{4} + 2 m \pi \right)^2 } - \cos\left( 
\frac{\pi}{4} + 2 m \pi \right) < 0, \\
 \widetilde{\mathcal{D}}_r\left( 
\frac{\pi}{2} + 2 m \pi \right) = \frac{ \bar{k}_p }{ \left( \frac{\pi}{2} + 2 m \pi \right)^2 } - \cos\left( \frac{\pi}{2} + 2 m \pi \right) > 0.
 \end{cases}
 \end{align}
Thus, there exists only one real root $\theta_{r, 2 m + 2} \in \left( \frac{\pi}{4} + 2 m \pi, \frac{\pi}{2} + 2 m \pi \right)$. 

 When $\theta \in \left( \frac{3 \pi}{2} + 2 n \pi, \frac{7 \pi}{4} + 2 n \pi \right)$, where $n = 0, 1, 2, \cdots$, we have
 \begin{align}
  \frac{\bar{k}_p}{\theta^3} > 0 > \sin\theta,
 \end{align}
 i.e., 
 \begin{align}
  \frac{d}{d \theta} \widetilde{\mathcal{D}}_r(\theta) = -\frac{2 \bar{k}_p}{\theta^3} + \sin \theta < 0,   
 \end{align}
 as well as 
 \begin{align}
 \begin{cases}{}
 \cos\left( \frac{3 \pi}{2} + 2 n \pi \right) < \frac{\bar{k}_p}{ \left(\frac{3 \pi}{2} + 2 n \pi \right)^2}, \\
 \cos\left( \frac{7 \pi}{4} + 2 n \pi \right) > \frac{\bar{k}_p}{ \left(\frac{7 \pi}{4} + 2 n \pi \right)^2},
 \end{cases}
 \end{align}
 i.e., 
 \begin{align}
  \begin{cases}
  \widetilde{\mathcal{D}}_r( \frac{3 \pi}{2} + 2 n \pi ) > 0, \\
   \widetilde{\mathcal{D}}_r( \frac{7 \pi}{4} + 2 n \pi ) < 0.
  \end{cases}    
 \end{align}
  Thus, there exists only one real root $\theta_{r, 2 n + 3} \in \left( \frac{3 \pi}{2} + 2 n \pi, \frac{7 \pi}{4} + 2 n \pi \right)$. 
 
Based on the above results, the positive real roots of~\eqref{eq:cosine-real-part-equation} satisfy the following: 
\begin{align} \label{eq:intervals-of-positive-roots-for-real-part-equation}
 0 &< \theta_{r, 1} < \frac{\pi}{4} < \theta_{r, 2} < \frac{\pi}{2} < \cdots \nonumber \\
 &< \left(k - 2 + \frac{3 - (-1)^k }{8} \right) \pi < \theta_{r,  k} \nonumber \\ 
 &< \left( k - 2 + \frac{5 - (-1)^k }{8} \right) \pi < \cdots,  
 \end{align}
 where $k \ge 2$. 
 
In addition, as a result of the symmetry of the roots of~\eqref{eq:cosine-real-part-equation}, the negative real roots satisfy the following: 
 \begin{align} \label{eq:intervals-of-negative-roots-for-real-part-equation}
 0 &> \theta_{r, -1} > - \frac{\pi}{4} > \theta_{r, -2} > - \frac{\pi}{2} > \cdots \nonumber \\
 &> - \left( k - 2 + \frac{ 3 - (-1)^k }{8} \right) \pi > \theta_{r, -k} \nonumber \\
&> - \left( k - 2 + \frac{ 5 - (-1)^k }{8} \right) \pi \nonumber \\
&> \cdots,    
 \end{align}
where $k \ge 2$ and $\theta_{r, -k}$ denotes the $k$-th largest negative root of~\eqref{eq:cosine-real-part-equation}.  \\

\textbf{Step 4:} On the basis of the roots given in~\eqref{eq:intervals-of-positive-roots-for-real-part-equation} and~\eqref{eq:intervals-of-negative-roots-for-real-part-equation}, we will determine the number of real roots of~\eqref{eq:cosine-real-part-equation} in each of the intervals $- 2 l \pi + \frac{\pi}{4} \le \theta \le 2 l \pi + \frac{\pi}{4}$ where $l = l_0, l_0 + 1, l_0 + 2, \cdots$, in which $l_0$ is a sufficiently large integer. This will allow us to invoke~{\rm \textbf{Lemma~\ref{lemma:real-roots}}}. 

First, when $l = 0$, there exist two roots: $\theta_{r, 1}$ and $\theta_{r, -1}$ in the interval $-\frac{\pi}{4} \le \theta \le \frac{\pi}{4}$. 
Second, substituting $k = 2 l + 1, l = 1, 2, \ldots$, into~\eqref{eq:intervals-of-positive-roots-for-real-part-equation}, we have  
 \begin{align}
 \left( k - 2 + \frac{ 5 - (-1)^k }{8} \right) \pi 
 &= 2 l \pi + \left( \frac{ 5 - (-1)^{2 l + 1} }{8} - 1 \right) \pi \nonumber \\
& < 2 l \pi + \frac{\pi}{4}. 
 \end{align} 
Substituting $ k = 2 l + 2$, we have 
 \begin{align}
\left( k - 2 + \frac{ 3 - (-1)^k }{8} \right) \pi  
&=2 l \pi + \left( \frac{ 3 - (-1)^{2 l + 2} }{8} \right) \pi \nonumber \\ 
&= 2 l \pi + \frac{\pi}{4}.  
 \end{align}
Thus, there exist exactly $(2 l + 1)$ positive roots in the interval $\left(0, 2 l \pi + \frac{\pi}{4} \right)$. 

On the other hand, when $\theta < 0$, substituting $k = 2 l + 1$, $l = 1, 2, \cdots$, into~\eqref{eq:intervals-of-negative-roots-for-real-part-equation}, we obtain 
 \begin{align} 
- \left( k - 2 + \frac{ 5 - (-1)^k }{8} \right) \pi &= -2 l \pi + \left( 1 - \frac{ 5 - (-1)^{2 l + 1} }{8} \right) \pi \nonumber \\ 
&= - 2 l \pi + \frac{\pi}{4},
 \end{align} 
which indicates that there exist exactly $(2 l + 1)$ negative roots in the interval $\theta \in \left(- 2 l \pi + \frac{\pi}{4}, 0 \right)$, $l = 1, 2, \ldots$. Thus, in the interval   
 \begin{align}
 - 2 l \pi + \frac{\pi}{4} \le \theta \le 2 l \pi + \frac{\pi}{4}, \ l = 1, 2, \ldots,
 \end{align}
 there exist exactly $(4 l + 2)$ real roots. Therefore, equation~\eqref{eq:cosine-real-part-equation} has only simple and real roots. \\   

\textbf{Roots of $\mathcal{D}_i(\theta) = 0$ or Equation~\eqref{eq:sine-imaginary-part-equation}}  \\

\textbf{Step 5:} We will first show that $\bar{\gamma} \le \frac{1}{2}$.




According to~\eqref{eq:quadratic-inequality-condition-1-cacc}, we have 
 \begin{align}
 \gamma \le \frac{1}{2 \tau_0},
 \end{align}
 based on which we obtain 
 \begin{align}
 \bar{\gamma} = \tau \gamma \le \tau_0 \gamma \le \frac{1}{2}. 
 \end{align}

 \textbf{Step 6:} We will show that there exists only one real root of~\eqref{eq:sine-imaginary-part-equation} in the interval $\left(0, \frac{\pi}{4} \right)$.

Since $\cos\theta > 0$ for all $\theta \in \left(0, \frac{\pi}{4} \right)$, we have
  \begin{align}
 \frac{d}{d\theta} \widetilde{\mathcal{D}}_i(\theta) = - \frac{\bar{\gamma}}{\theta^2} - \cos\theta < 0.
  \end{align}
 Thus, $\widetilde{\mathcal{D}}_i(\theta)$ is monotonically decreasing in this interval; in addition, for all $\theta \in \left(0, \frac{\pi}{4} \right)$, it also holds that 
 \begin{align}
 \widetilde{\mathcal{D}}_i(0) &= \lim\limits_{\theta \to 0} \frac{\bar{\gamma}}{\theta} - \sin(0) > 0, \\
  \widetilde{\mathcal{D}}_i\left(\frac{\pi}{4} \right) &= \frac{\bar{\gamma}}{ \left( \frac{\pi}{4} \right) } - \sin\left(\frac{\pi}{4} \right) \le \frac{2}{\pi} - \frac{\sqrt{2}}{2} < 0.
 \end{align}
  Thus, there exists only one real root $\theta_{i, 2} \in \left(0, \frac{\pi}{4} \right)$. \\

 \textbf{Step 7:} We will show that there exists no real root in the interval $\left( \frac{\pi}{4}, \frac{3 \pi}{4} \right)$ and only one real root in the interval $\left(\frac{3 \pi}{4}, \pi \right)$ of~\eqref{eq:sine-imaginary-part-equation}. 

 When $\theta \in \left( \frac{\pi}{4}, \frac{3 \pi}{4} \right)$: Since it 
  holds in this interval that 
 $\sin\theta > \sin\left(\frac{\pi}{4}\right) > \frac{\bar{\gamma}}{ \left( \frac{\pi}{4} \right) } > \frac{\bar{\gamma}}{\theta}$, then $\widetilde{\mathcal{D}}_i(\theta) < 0$ in this interval. Thus, there is no real root in the interval $\left(\frac{\pi}{4}, \frac{3 \pi}{4} \right)$. 
 
 When $\theta \in \left( \frac{3 \pi}{4}, \pi \right)$: In this interval, we have  
 \begin{align}
 - \frac{\bar{\gamma}}{\theta^2} > - \frac{ \bar{\gamma} }{ \left( \frac{3 \pi}{4} \right)^2 } > - \frac{ \frac{1}{2} }{ \left( \frac{3 \pi}{4} \right)^2 } > - \frac{\sqrt{2}}{2}, \\
 \cos\theta < \cos\left(\frac{3 \pi}{4} \right) = - \frac{\sqrt{2}}{2}. 
 \end{align} 
 Thus, 
 \begin{align}
 \frac{d}{d \theta} \widetilde{\mathcal{D}}_i(\theta) = - \frac{\bar{\gamma}}{\theta^2} - \cos \theta > 0,
 \end{align}
 i.e., $\widetilde{\mathcal{D}}_i(\theta)$ is monotonically increasing in this interval; it also holds that     
  \begin{align}
   \widetilde{\mathcal{D}}_i\left( \frac{3 \pi}{4} \right) &= \frac{\bar{\gamma}}{ \left(\frac{3 \pi}{4} \right) } - \sin\left(\frac{3 \pi}{4} \right) < 0, \\
    \widetilde{\mathcal{D}}_i(\pi) &= \frac{\bar{\gamma}}{\pi} - \sin(\pi) > 0.
  \end{align}
 Thus, there exists only one real root $\theta_{i, 3} \in \left(\frac{3 \pi}{4}, \pi \right)$. 
 
 Proceeding with the above procedure and using the periodicity of $\sin\theta$, we can conclude the following: 
 
 When $\theta \in \left(2 m \pi, \frac{\pi}{4} + 2 m \pi \right)$, $m = 0, 1, 2, \cdots$, it holds that 
 \begin{align}
 - \frac{\bar{\gamma}}{\theta^2} < 0 < \cos \theta, 
 \end{align}
 i.e., 
 \begin{align}
 \frac{d}{d \theta} \widetilde{\mathcal{D}}_i(\theta) = - \frac{\bar{\gamma}}{\theta^2} - \cos \theta < 0,    
 \end{align}
 then $\widetilde{\mathcal{D}}_i(\theta)$ is monotonically decreasing in this interval; in addition, 
 \begin{align}
 \begin{cases}{}
 \sin( 2 m \pi ) < \frac{\bar{\gamma}}{ 2 m \pi }, \\
 \sin \left( \frac{\pi}{4} + 2 m \pi \right) > \frac{ \bar{\gamma} }{ \frac{\pi}{4} + 2 m \pi },
 \end{cases}
 \end{align}
 i.e., 
 \begin{align}
  \begin{cases}
  \widetilde{\mathcal{D}}_i(2 m \pi) = \frac{\bar{\gamma}}{2 m \pi} - \sin(2 m \pi) > 0, \\
  \widetilde{\mathcal{D}}_i\left( \frac{\pi}{4} + 2 m \pi \right) = \frac{ \bar{\gamma} }{ \frac{\pi}{4} + 2 m \pi } - \sin\left( \frac{\pi}{4} + 2 m \pi \right) < 0.
  \end{cases}
 \end{align}
Thus, there exists only one real root $\theta_{i, 2 m + 2} \in \left( 2 m \pi, \frac{\pi}{4} + 2 m \pi \right)$. 
 
 When $\theta \in \left( \frac{3 \pi}{4} + 2 n \pi, \pi + 2 n \pi \right)$, $n = 0, 1, 2, \cdots$, it holds that  
 \begin{align}
  \cos \theta < - \frac{\sqrt{2}}{2} < - \frac{\bar{\gamma}}{\theta^2},
 \end{align}
 then 
 \begin{align}
 \frac{d}{d \theta} \widetilde{\mathcal{D}}_i(\theta) = - \frac{\bar{\gamma}}{\theta^2} - \cos \theta > 0,
 \end{align}
 that is, $\widetilde{\mathcal{D}}_i(\theta)$ is monotonically increasing in this interval; as well as 
 \begin{align} 
 \begin{cases}{}
 \sin\left( \frac{3 \pi}{4} + 2 n \pi \right) > \frac{ \bar{\gamma} }{ \frac{3 \pi}{4} + 2 n \pi }, \\
 \sin( \pi + 2 n \pi ) < \frac{ \bar{\gamma} }{ \pi + 2 n \pi },
 \end{cases}
 \end{align}
 i.e., 
  \begin{align}
   \begin{cases}
   \widetilde{\mathcal{D}}_i\left(\frac{3 \pi}{4} + 2 n \pi \right) = \frac{\bar{\gamma}}{ \frac{3 \pi}{4} + 2 n \pi } - \sin\left( \frac{3 \pi}{4} + 2 n \pi \right) < 0, \\
   \widetilde{\mathcal{D}}_i(\pi + 2 n \pi) = \frac{\bar{\gamma}}{\pi + 2 n \pi} - \sin(\pi + 2 n \pi) > 0.
   \end{cases}
  \end{align} 
  Thus, there exists only one real root $\theta_{i, 2 n + 3} \in \left( \frac{3 \pi}{4} + 2 n \pi, \pi + 2 n \pi \right)$.

Then, the non-negative roots of~\eqref{eq:sine-imaginary-part-equation} satisfy the following: 
\begin{align} \label{eq:intervals-of-nonnegative-roots-for-imaginary-part-equation}
 0 &= \theta_{i, 1} < \theta_{i, 2} < \frac{\pi}{4} < \cdots \nonumber \\
 &< ( k - 2 ) \pi + \left( \frac{ (-1)^k - 1 }{8} \right) \pi < \theta_{i, k} \nonumber \\
&< ( k - 2 ) \pi + \left( \frac{ (-1)^k + 1 }{8} \right) \pi < \cdots,
 \end{align}
  where $k \ge 2$. Further, due to the symmetry of the roots of~\eqref{eq:sine-imaginary-part-equation}, the negative roots can be assigned in a descending order as follows:
\begin{align} \label{eq:intervals-of-negative-roots-for-imaginary-part-equation}
 0 &> \theta_{i, -1} > - \frac{\pi}{4} > \cdots \nonumber \\
 &>  - ( k - 1) \pi - \left( \frac{ (-1)^{k + 1} - 1 }{8} \right) \pi > \theta_{i, -k} \nonumber \\ 
 &> - ( k - 1 ) \pi - \left( \frac{ (-1)^{k + 1} + 1 }{8}  \right) \pi > \cdots,   
 \end{align}
 where $k \ge 1$. \\

\textbf{Step 8:} Based on~\eqref{eq:intervals-of-nonnegative-roots-for-imaginary-part-equation} and~\eqref{eq:intervals-of-negative-roots-for-imaginary-part-equation}, we will determine the number of real roots of~\eqref{eq:sine-imaginary-part-equation} in each of the intervals $-2 l \pi + \frac{\pi}{4} \le \theta \le 2 l \pi + \frac{\pi}{4}$ where $l = l_0, l_0 + 1, l_0 + 2, \cdots$, and $l_0$ is a sufficiently large integer. First, substituting $k = 2 l + 2$, $l \ge 1$, into~\eqref{eq:intervals-of-nonnegative-roots-for-imaginary-part-equation}, we have 
 \begin{align}
 ( k - 2 ) \pi + \left( \frac{ (-1)^k + 1 }{8} \right) \pi &= 2 l \pi + \left( \frac{ (-1)^{ 2 l + 2 } + 1 }{8} \right) \pi \nonumber \\
 &= 2 l \pi + \frac{\pi}{4}.
 \end{align}
On the other hand, for $\theta < 0$, substituting $k = 2 l$ into~\eqref{eq:intervals-of-negative-roots-for-imaginary-part-equation} leads to
 \begin{align}
 - ( k - 1 ) \pi - & \left( \frac{ (-1)^{k + 1} + 1 }{8} \right) \pi \nonumber \\
 &= - ( 2 l - 1 ) \pi - \left( \frac{ (-1)^{2 l + 1} + 1 }{8} \right) \pi \nonumber \\
 &= - (2 l - 1 ) \pi 
 >- 2 l \pi + \frac{\pi}{4}.
 \end{align}
In addition, substituting $k = 2 l + 1$ into~\eqref{eq:intervals-of-negative-roots-for-imaginary-part-equation} yields 
 \begin{align}
 - ( k - 1) \pi - &\left( \frac{ (-1)^{k + 1} - 1 }{8} \right) \pi \nonumber \\
 &= - ( 2 l + 1 - 1 ) \pi - \left( \frac{ (-1)^{2 l + 2} - 1 }{8} \right) \pi \nonumber \\
 &= -2 l \pi < - 2 l \pi + \frac{\pi}{4},
 \end{align}
 which indicates that when 
 \begin{align}
 - 2 l \pi + \frac{\pi}{4} \le \theta \le 2 l \pi + \frac{\pi}{4}, l = 1, 2, \cdots, 
 \end{align} 
 there exist exactly $(4 l + 2)$ real roots of~\eqref{eq:sine-imaginary-part-equation}. Thus, equation~\eqref{eq:sine-imaginary-part-equation} has only simple and real roots. \\

\textbf{Step 9:} 
Now, we will show that the roots of~\eqref{eq:cosine-real-part-equation} and~\eqref{eq:sine-imaginary-part-equation} interlace. 

\begin{figure*}[htb!]
\centering{\includegraphics[scale=0.75]{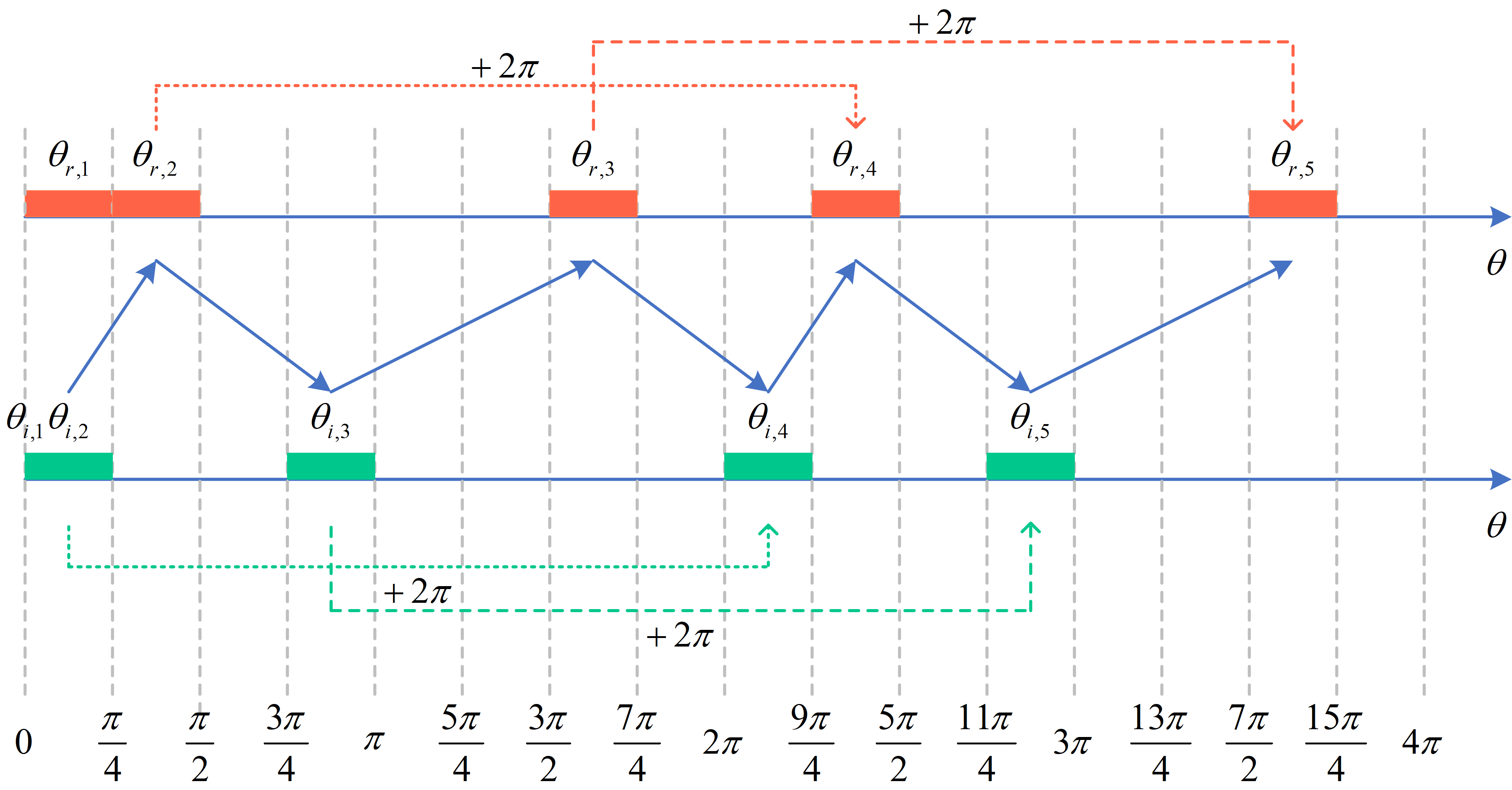}}
\caption{An illustration of the distribution, periodicity and interlacing property of the non-negative roots of $\mathcal{D}_r(\theta) = 0$ and $\mathcal{D}_i(\theta) = 0$.
\label{fig:interlacing-illustration}}
\end{figure*}

As shown in Figure~\ref{fig:interlacing-illustration}, for $k\ge 2$, this is straightforward because
 \begin{align}
 \frac{ (-1)^k + 1 }{8} \le \frac{ 3 - (-1)^k }{8} < 1 + \frac{ (-1)^{k + 1} - 1 }{8},
 \end{align} 
which means that $\theta_{i, k} < \theta_{r, k} < \theta_{i, k + 1}$. Thus,
 \begin{align}
 \theta_{i, 2} < \theta_{r, 2} < \theta_{i, 3} < \theta_{r, 3} < \theta_{i, 4} < \cdots. 
 \end{align}

Next, we will show that the first two roots interlace, i.e., $\theta_{r, 1} < \theta_{i, 2}$. For this purpose, we start with  $\cos\theta > 1 - \frac{\theta^2}{2}$
 for all $\theta > 0$. Then the roots of the equation 
 \begin{align}
 \frac{\bar{k}_p}{\theta^2} = 1 - \frac{\theta^2}{2},
 \end{align}
 are given by 
 \begin{align}
 \hat{\theta}_{r, 1} = \sqrt{ 1 - \sqrt{ 1 - 2 \bar{k}_p } }, \hat{\theta}_{r, 2} = \sqrt{ 1 + \sqrt{ 1 - 2 \bar{k}_p } }.    
 \end{align}
 Then we have $\theta_{r, 1} < \hat{\theta}_{r, 1}$.         
On the other hand, $\sin\theta < \theta$ for all $\theta > 0$. Then, the positive root of 
 \begin{align}
 \theta = \frac{\bar{\gamma}}{\theta},
 \end{align}
 is given by $\hat{\theta}_{i, 2} = \sqrt{\bar{\gamma}}$. 
Thus, $\theta_{i, 2} > \hat{\theta}_{i, 2}$.

Now, to show $\theta_{r, 1} < \theta_{i, 2}$, it suffices to show $\hat{\theta}_{r,1} < \hat{\theta}_{i,2}$, that is,  
\begin{align} \label{eq:hat-theta-r-1-hat-theta-i-2-inequality-condition}
1 - \sqrt{ 1 - 2 \tau^2 k_p } < \bar{\gamma}. 
\end{align}
To prove this inequality, 
first, the following is true:
\begin{align}
 ( 1 - 2 \bar{k}_p ) ( 2 \bar{k}_p + 2 )^2 = - ( 2 \bar{k}_p )^3 - 3 ( 2 \bar{k}_p )^2 + 4 < 4.   
\end{align}
 Taking the square root on both sides of the above inequality yields 
 \begin{align}
  \sqrt{ 1 - 2 \bar{k}_p } ( 2 \bar{k}_p + 2 ) < 2.    
 \end{align}
 Multiplying both sides by $\sqrt{ 1 - 2 \bar{k}_p }$ results in 
 \begin{align}
  ( 1 - 2 \bar{k}_p ) ( 2 \bar{k}_p  + 2 ) < 2 \sqrt{ 1 - 2 \bar{k}_p },   
 \end{align}
 which can be rewritten as 
 \begin{align}
 ( 2 \bar{k}_p )^2 > \left( 1 - \sqrt{ 1 - 2 \bar{k}_p } \right)^2,    
 \end{align}
 from which, we obtain 
\begin{align} \label{eq:condition-1-for-left-inequality}
 1 - \sqrt{ 1 - 2 \tau^2 k_p } < 2 \tau^2 k_p.
 \end{align} 
 In addition, from~\eqref{eq:tau-square-kp-inequality-1}, we have
 \begin{align}
 2 \tau^2 k_p < \frac{ (1 - k_a) ( 1 + k_a )^2 }{4} < 1 - k_a.
 \end{align}
Thus, 
 \begin{align}
 2 \tau^2 k_p < \sqrt{ 2 \tau^2 k_p ( 1 - k_a ) } < \sqrt{ 2 \tau^2 k_p ( 1 - k_a ) + \tau^2 k_v^2 },
 \end{align} 
 i.e., 
 \begin{align}
 2 \tau^2 k_p &< \tau \sqrt{ 2 k_p ( 1 - k_a ) + k_v^2 } \nonumber \\
 &\le \bar{\gamma}.
 \end{align}
 Using~\eqref{eq:condition-1-for-left-inequality}, we have \begin{align}
 1 - \sqrt{ 1 - 2 \tau^2 k_p } < \tau \sqrt{ 2 k_p ( 1 - k_a ) + k_v^2 } \le \bar{\gamma}. 
 \end{align}
Therefore, $\theta_{r,1} < \theta_{i,2}$.

For the numerical values considered in the simulations, $\bar{k}_p = 0.0150$ and $\bar{\gamma} = 0.3710$, the graphs of  $\frac{\bar{k}_p}{\theta^2}$ vs. $\cos\theta$, and $\frac{\bar{\gamma}}{\theta}$ vs. $\sin\theta$ are provided in Figure~\ref{fig:illustration-of-curves} to numerically show the interlacing property when $\theta \in (0,2\pi)$.
\begin{figure}[!htb]
\centering{\includegraphics[scale=0.6]{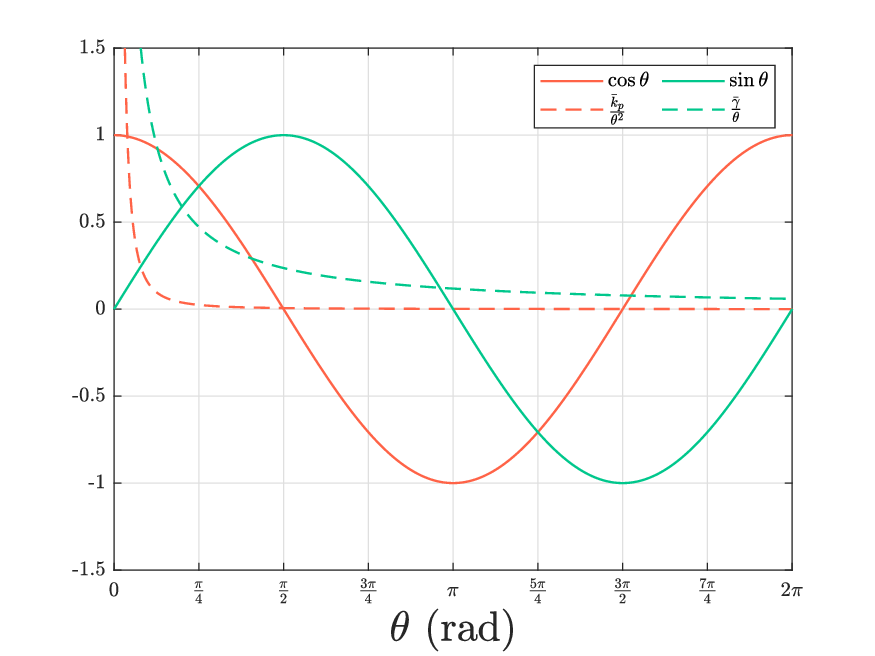}}
\caption{An illustration of $\frac{\bar{k}_p}{\theta^2}$ vs. $\cos \theta$ and $\frac{\bar{\gamma}}{\theta}$ vs. $\sin \theta$.
\label{fig:illustration-of-curves}}
\end{figure} \\

\noindent \textbf{Proof of Statement (B)} 

It follows from~\eqref{eq:D-ast-j-omega} that 
\begin{align}
 \mathcal{D}_r^{\prime}(\omega) &= - 2 \tau^2 \omega \cos(\tau \omega) + \tau^3 \omega^2 \sin(\tau \omega), \\
 \mathcal{D}_i^{\prime}(\omega) &= \tau^2 \gamma - 2 \tau^2 \omega \sin(\tau \omega) - \tau^3 \omega^2 \cos(\tau \omega).
\end{align}
 From this, we can derive the following:
 \begin{align} \label{eq:internal-stability-condition-2}
 &~~~\mathcal{D}_i^{\prime}(\omega) \mathcal{D}_r(\omega) - \mathcal{D}_i(\omega) \mathcal{D}_r^{\prime}(\omega) \nonumber \\
  &= \tau^4 \gamma k_p + \tau^5 \omega^4 - \tau^3 \omega ( 2 \tau k_p + \tau^2 \gamma \omega^2 ) \sin(\tau \omega) \nonumber \\
  &\quad + \tau^4 \omega^2 ( \gamma - \tau k_p ) \cos(\tau \omega).
 \end{align}
 Substituting $\omega = 0$, we obtain 
 \begin{align}
 \mathcal{D}_i^{\prime}(\omega) \mathcal{D}_r(\omega) - \mathcal{D}_i(\omega) \mathcal{D}_r^{\prime}(\omega) = \tau^2 \gamma k_p > 0.    
 \end{align}
This completes the proof for~{\rm \textbf{Theorem~\ref{theorem:the-inclusion-of-internal-stability-in-string-stability-cacc-case}}}. 
\end{proof}


 \section{} \label{Appendix_B}
 \begin{proof}[Proof for Theorem~\ref{theorem:the-inclusion-of-internal-stability-in-string-stability-cacc-plus-case}]
 Note from~\eqref{eq:string-stability-condition-cacc-plus-case-with-identical-control-gains}, one can define 
 \begin{align}
 H_r(s; \tau) = \frac{\mathcal{N}_r(s)}{\mathcal{D}_r(s)},
 \end{align}
 where $\mathcal{N}_r(s) = \tilde{k}_a s^2 + \tilde{k}_v s + \tilde{k}_p$, $\mathcal{D}_r(s) = s^2 e^{\tau s} + ( \tilde{k}_v + \tilde{k}_p \tilde{h}_w ) s + \tilde{k}_p$. Then the inter-vehicular spacing error propagation equation~\eqref{eq:inter-vehicular-spacing-error-propagation-equation-cacc-plus-case} with identical control gains, i.e. $k_{aj} = k_a$, $k_{vj} = k_v$, $k_{pj} = k_p$, is internally stable iff the polynomial $\mathcal{D}_r(s)$ satisfies the conditions in~{\rm \textbf{Lemma~\ref{lemma:internal-stability}}}. Comparing $H_r(s; \tau)$ with $H(s; \tau)$ as defined in~\eqref{eq:inter-vehicular-spacing-error-propagation-equation-cacc-case}, the conditions on $\tilde{k}_a$, $\tilde{k}_v$, $\tilde{k}_p$ and $\tilde{h}_w$ derived in~{\rm \textbf{Theorem~\ref{theorem:string-stability-condition-ka-range-cacc-plus-case}}} and~{\rm \textbf{Theorem~\ref{theorem:string-stability-condition-cacc-plus-case}}}, correspond to those for $k_a$, $k_v$, $k_p$ and $h_w$ derived in~{\rm \textbf{Theorem~\ref{theorem:string-stability-condition-ka-range}}} and~{\rm \textbf{Theorem~\ref{theorem:string-stability-condition}}}. Therefore, following the same procedure as in~{\rm \textbf{Theorem~\ref{theorem:the-inclusion-of-internal-stability-in-string-stability-cacc-case}}}, internal stability of polynomial $\mathcal{D}_r(s)$ can be ensured by the derived string stability condition for $H_r(s; \tau)$. Therefore, the proof for~{\rm \textbf{Theorem~\ref{theorem:the-inclusion-of-internal-stability-in-string-stability-cacc-plus-case}}} is completed. 
 \end{proof}



	\bibliography{reference}
	\bibliographystyle{plain}
\end{document}